\documentclass[11pt]{article}
\pdfoutput=1
\usepackage[utf8]{inputenc}
\usepackage{jheppub}

\usepackage{graphicx}
 \usepackage{subcaption}
\usepackage{amsmath, amsthm, amssymb}

\usepackage{epsfig,latexsym,cancel}
\usepackage{amssymb,amsmath}
\usepackage{epstopdf}
\usepackage{hyperref, graphicx, slashed, subcaption}
\usepackage{multirow}
\usepackage[usenames, dvipsnames]{xcolor}
\usepackage{array}
\usepackage{booktabs}

\usepackage{epsfig,latexsym,cancel}
\usepackage{amssymb,amsmath}
\usepackage{epstopdf}
\usepackage{multirow}
\usepackage{array}
\usepackage{booktabs}
\usepackage{subcaption}

\newtheorem{theorem}{\sf THEOREM}
\newtheorem{lemma}{\sf LEMMA}

\newcommand{\brr}[1]{\left(#1\right)}
\newcommand{\srr}[1]{\left[#1\right]}

\newcommand{\sq}{\sin^2\!\left(\tfrac{\theta}{2}\right)}
\newcommand{\cq}{\cos^2\!\left(\tfrac{\theta}{2}\right)}

\title{Entanglement, Yang-Mills, and the Scattering Matrix as an SU(N)-equivariant Kernel}

\author{Kun-Feng Lyu,}
\author{Rahul Muraleedharan,}
\author{Kuver Sinha}
\affiliation{Homer L. Dodge Department of Physics and Astronomy, University of Oklahoma, Norman, OK 73019, USA}

\abstract{
We study two-body scattering as an SU(N)-equivariant map acting on tensor-product representation spaces and analyze the entanglement generated by the $S$-matrix.  
This representation-theoretic perspective separates group structure from dynamics: the decomposition of $R\!\otimes\!R'$ fixes the invariant operator algebra and therefore the qualitative entangling power of the process.  For particles in the fundamental representation, $\mathrm{End}_{\mathrm{SU}(N)}(N\!\otimes\!N)=\mathrm{Span}\{\mathbb{I},\mathbb{S}\}$, so only the identity and swap directions preserve separability, whereas generic combinations generate entanglement. Adjoint-adjoint scattering involves a larger invariant algebra involving $d$-tensors and is intrinsically entangling. In Yang-Mills theory one can use color-kinematics duality to show that the color kernel lies on a fixed ray of this operator space, yielding a universal maximum of the outgoing entanglement for scattering at right angles, $E_\star^{(2)}=\tfrac{3}{4}$ for $SU(2)$ and $E_\star^{(3)}\simeq0.91$, independent of kinematics. Dimension-six operators preserve this universality, while dimension-eight deformations populate new color sectors and shift $E_\star^{(N)}$, suggesting that entanglement in color space  functions as a tomographic probe of effective operators. In helicity space, requiring maximally entangled inputs to scatter into maximally entangled outputs uniquely selects the Yang-Mills quartic coupling and enforces the color Jacobi identity, restating the on-shell Ward constraints as conditions on entanglement preservation. Our results suggest that the information-theoretic viewpoint unifies algebraic, geometric, and dynamical aspects of scattering. 
}

\begin{document}

\maketitle

\section{Introduction}

Scattering amplitudes encode the dynamical and algebraic structure of quantum field theory.  Beyond their conventional role in computing cross sections, they provide an arena where symmetry, representation theory, and quantum information intersect.  
In recent years, the entanglement properties of scattering processes have been used to extract group-theoretic and geometric information directly from the $S$-matrix, suggesting that certain symmetry patterns can be understood as statements about the separability or entangling power of the underlying operator\footnote{There has been much recent activity in applications of entanglement to collider physics, which is not the focus of our attention. For representative papers along this direction, we refer to~\cite{Fabbrichesi:2021npl,Severi:2021cnj,Afik:2022kwm,Aguilar-Saavedra:2022uye, Cheng:2024btk, Afik:2020onf, Dong:2023xiw,Han:2023fci, Fabbrichesi:2022ovb, Aoude:2023hxv, Maltoni:2024tul,Duch:2024pwm,Fabbrichesi:2025ywl, Aoude:2022imd,Severi:2022qjy,Gu:2025ijz}.}~\cite{Beane:2018oxh,Low:2021ufv,McGinnis:2025brt}.

The purpose of this paper is to develop a unified framework in which these ideas are made fully representation-theoretic. Taking the external lines to carry representations $R$ and $R'$ of SU(N), so the initial Hilbert space is 
$\mathcal H_R\otimes\mathcal H_{R'}$, 
we treat the two-body scattering matrix as an SU(N)-equivariant map
\begin{equation}
     K:\; \in\ \mathrm{End}^{\rm SU(N)}(\mathcal H_R\otimes\mathcal H_{R'})\,,
  \qquad [K,\,U\otimes U]=0,\quad U\in \mathrm{SU}(N),
\end{equation}
and study the entanglement properties of its action on product states. $K$ thus lives in the commutant algebra generated by the projectors $P_{\mathcal R}$ onto irreps ${\mathcal R}\subset R\otimes R'$.
This approach disentangles the purely algebraic content of the amplitude from its dynamical details: once the tensor-product decomposition of the external representations is known, the space of SU(N)-equivariant endomorphisms $\mathrm{End}^{\rm SU(N)}(\mathcal H_R\otimes\mathcal H_{R'})$, and therefore the possible patterns of entanglement are completely fixed by group theory. The physical dynamics only select particular coefficient functions within this invariant operator space. The dependence of the entanglement on the group representation and on the coefficients of the invariant tensors makes it a sensitive probe of both algebraic and dynamical structure.

The subsequent specialization to gluon scattering in Yang-Mills theory simply amounts to choosing the adjoint representation and a particular linear combination of invariant tensors, which, in this case, are the well-known double-$f$ color structure. For adjoint scattering, the presence of the singlet and $d$-tensor sectors in the invariant operator algebra unavoidably mix symmetric and antisymmetric channels, so that even initially separable gluon states become entangled after scattering. Adjoint-adjoint scattering is therefore intrinsically an \emph{entangler}, a statement that depends only on representation theory.

The kinematics in the case of Yang-Mills  appears in the form of the usual color-kinematics duality~\cite{Bern:2008qj}, leading to a specific form of the coefficient functions within the invariant operator space.  Notably, the coefficient functions are determined entirely by the group structure and the scattering angle, and do not rely on the kinematic scattering amplitude. We find that the outgoing entanglement reaches a universal maximum at right angles, $\theta=\pi/2$, independent of the specific choice of external color states.
This universality follows from the algebraic structure of the invariant tensors rather than from kinematics. The value of this maximum outgoing entanglement is therefore a group invariant. We explicitly work it out in the cases of SU(2) and SU(3).

For fundamental scattering, on the other hand, the space of invariant operators is two-dimensional, spanned by the identity and swap operators. The only separability-preserving rays in this space correspond to the identity and swap directions; for any other choice of ray in this space $K$ generates entanglement in the outgoing state. The absence of singlet and $d$-tensor projections implies that fundamental-fundamental scattering is ``minimally entangling" in the sense that no additional invariant tensors exist to increase its entangling power. These results hold regardless of whether the SU(N) is a global or gauge symmetry.

Given that for adjoint scattering in Yang-Mills the outgoing entanglement reaches a universal maximum at $\theta=\pi/2$ that only depends on the algebraic structure of the invariant tensors, one can utilize this as a probe of operators that violate this structure. We study, in particular, dimension six and dimension eight operators. Color-kinematics duality continues to hold after incorporating dimension six operators \cite{Dixon:1993xd,Dixon:2004za,Broedel:2012rc}, which implies no alteration of the double-$f$ structure; dimension six operators, therefore, leave the entanglement unchanged. Dimension eight operators, on the other hand,  introduce new $dd$ product sectors, altering the universal value obtained at dimension four, by amounts that depend on the cut-off scale.

In the second part of the paper, we  study bipartite entanglement in the helicity Hilbert space for $2\!\to\!2$ scattering of massless vector bosons~\cite{Nunez:2025dch}. Here, the scattering operator acts on the tensor product of one-particle helicity spaces rather than color spaces. We perform purely \textit{formal} deformations\footnote{These  are purely formal in the sense that they are undertaken to check how well entanglement detects such unphysical deformations.} away from the Yang-Mills Lagrangian, allowing a cubic term $AA\partial A$ with general coefficients $f_{abc}$ and a local quartic with an a priori free coefficient $\kappa$. It is well known that demanding on-shell gauge invariance (the Ward identities) together with locality fixes the cubic coefficients $f_{abc}$ to be antisymmetric, determines the coefficient in the quartic term to be the Yang-Mills coefficient ($\kappa=1$), and enforces the  Jacobi relation among the coefficients $f_{abc}$, so that one obtains the usual $\mathcal L_{\rm YM} \;=\; -\frac{1}{4}\, F^a_{\mu\nu} F^{a\,\mu\nu}$ (we call this set of conditions the ``YM locus") \cite{Cheung:2017pzi}. We ask whether the same information leading to the  Yang-Mills locus can be recast  as information-theoretic statements about  the $2\!\to\!2$ scattering map.

  Schematically, the approach consists in preparing  an initial two-gluon state $|\mathrm{in}\rangle$ that is maximally entangled (MaxE) in the $\{\!+,-\!\}$ helicity basis on the two external legs and studying the normalized entanglement of the outgoing state $|\mathrm{out}\rangle=M\cdot|\mathrm{in}\rangle$, with $M$ the helicity amplitude matrix.  We quantify bipartite entanglement by the linear entropy which equals $1$ if and only if the out-state is maximally entangled (MaxE) and vanishes for products. For normalized $|\mathrm{out}\rangle$ the entanglement is basis independent on the two helicity legs and insensitive to the overall scaling of $M$. Our results are as follows. 
 
 $(i)$ MaxE$\to$MaxE in Yang-Mills: 
We find that for all scattering angles and for all maximally entangled inputs, the final state is always maximally entangled in Yang-Mills. 
Algebraic checks show that the final state remains maximally entangled due to the Maximally Helicity Violated (MHV) property. 

 $(ii)$ Departure from Yang-Mills leads to violation of MaxE$\to$MaxE: Under formal deviations from $\kappa =1$ and the color Jacobi identity $c_s + c_t + c_u = 0$, and assuming that the Ward identity no longer applies (the external polarization vector can be shifted: $\epsilon_\mu \rightarrow \epsilon_\mu + \xi\, p_\mu$), we calculate the  change of the entanglement of the final state. We find  that only the Yang-Mills locus respects the MaxE$\to$MaxE behavior.

There has been much recent work probing the intriguing correlation between symmetry and entanglement, and we end the Introduction by mentioning the broad contours of the literature. The fact that suppression of  entanglement in 2-to-2 scattering can lead to the enhanced  symmetries was pointed out by ~\cite{Beane:2018oxh,Low:2021ufv,McGinnis:2025brt,McGinnis:2025xgt,McGinnis:2025iab}. The first discussion on the correlation between  entanglement and emergent symmetries was in the context of low energy baryon scattering systems~\cite{Beane:2018oxh}. Subsequently, other scenarios were studied, including meson scattering~\cite{Hu:2024hex},  octet and decuplet baryons scattering~\cite{Liu:2022grf,Liu:2023bnr}, lowest-lying spin-3/2 baryons~\cite{Hu:2025jne} scattering, 
quantum electrodynamics scattering processes~\cite{Blasone:2024dud, Blasone:2024jzv, Blasone:2025ddi, Blasone:2025tor},
Higgs scattering in two-Higgs-doublet models (2HDMs)~\cite{Carena:2023vjc,Kowalska:2024kbs,Chang:2024wrx,Busoni:2025dns}, scattering mediated by heavy particles~\cite{Sou:2025tyf}, forward scattering~\cite{Kowalska:2025qmf}, crossing symmetry in the scattering~\cite{McGinnis:2025xgt},the area law in particle scattering~\cite{Low:2024mrk,Low:2024hvn},
understanding the  pattern of flavor symmetry and weak mixing angles in the Standard Model~\cite{Thaler:2024anb,Liu:2025bgw} and the dynamics of electroweak phase transition~\cite{Liu:2025pny}. Recent discussions of entanglement in gauge theories occur in~\cite{Cervera-Lierta:2017tdt,Cheung:2020uts,Nunez:2025dch,Gargalionis:2025iqs, Nunez:2025xds}.

The rest of the paper is organized as follows. 
Sec.~\ref{sec:2_SU(N)map} lays out the basic setup of the entanglement power in the Hilbert space with adjoint representation.
Then we explicitly computed the entanglement in Yang-Mills in Sec.~\ref{YMkins2}. The entanglement in the helicity space is discussed in Sec.~\ref{sec:entanle_hel}. Finally we summarize and conclude in Sec.~\ref{sec:conclusion}.

\section{Scattering Matrix as an SU(N)-equivariant Map}\label{sec:2_SU(N)map}


\subsection{Motivation and General Scope}

Throughout this section we treat the two-to-two scattering matrix as a linear map between tensor-product representation spaces of a compact group SU(N).
All statements will rely only on the fact that the scattering operator is
SU(N)-equivariant, i.e., it commutes with the diagonal group action on the two-particle Hilbert space.
No assumption is made about whether the symmetry arises from a local gauge invariance or from a global flavor symmetry, nor about the microscopic dynamics.

This perspective emphasizes that many features of the entanglement structure of scattering processes follow purely from representation theory:
once the tensor-product decomposition $R\otimes R'$ is known,
the space of SU(N)-equivariant endomorphisms $\mathrm{End}_{SU(N)}(R\otimes R')$, and hence the possible form of the scattering operator, are fixed. The physical dynamics only select particular coefficient functions within this invariant operator space.

\subsection{Entanglement Formalism and the Scattering Operator}

We denote the initial state by \(|\mathrm{in}\rangle\) and the final state by
\(|\mathrm{out}\rangle = S\,|\mathrm{in}\rangle\),
where \(S= \mathbb{I}+ iT\) is the unitary S-matrix.
For a \(2\!\to\!2\) process
\begin{equation}
     g_a(p_1)+g_b(p_2)\rightarrow g_c(p_3)+g_d(p_4),
\end{equation}
the reduced amplitude
\begin{equation}
M_{cd,ab}(p_1,p_2\!\to\!p_3,p_4)
\end{equation}
serves as the matrix elements of a linear map
acting on the product Hilbert space
\(\mathcal{H}_1\!\otimes\!\mathcal{H}_2\).
Suppressing kinematic delta functions,
we define the corresponding operator
\begin{equation}
  K(\Theta)_{cd,ab}
   := M_{cd,ab}(p_1,p_2\!\to\!p_3,p_4),
   \qquad
   \Theta=\{\text{helicities, momenta, scattering angle}\}.
\end{equation}
In what follows, $K$ will be treated as an abstract SU(N)-equivariant map,
without specifying the functional dependence on~\(\Theta\).
The operator acts on product basis vectors \(|a,b\rangle:=|a\rangle\!\otimes\!|b\rangle\) as
\begin{equation}\label{eq:defK}
  K(\Theta): |a,b\rangle \xrightarrow{K_{cd,ab}(\Theta)} \,|c,d\rangle .
\end{equation}

Given an input product vector
\(|\psi_{\mathrm{in}}\rangle = |u\rangle \!\otimes \!|v\rangle\),
with components \(u_a,v_b\) in the chosen basis of the representation \(R\),
the outgoing state coefficients are collected into the matrix
\begin{equation}
\mathcal{N}_{cd}=K_{cd,ab}u_a v_b .
\end{equation}
The normalized reduced density matrix on one leg is
\begin{equation}
\rho_R=\frac{\mathcal{N}\mathcal{N}^{\dagger}}{\mathrm{Tr}[\mathcal{N}\mathcal{N}^{\dagger}]},
\end{equation}
and the (normalized) linear entropy measuring bipartite entanglement is
\begin{equation}
E=\frac{d_R}{d_R-1}\Bigl(1-\mathrm{Tr}\,\rho_R^{\,2}\Bigr)
   =\frac{d_R}{d_R-1}
     \Bigl(
     1-\frac{\mathrm{Tr}[(\mathcal{N}\mathcal{N}^{\dagger})^2]}
             {(\mathrm{Tr}[\mathcal{N}\mathcal{N}^{\dagger}])^2}
     \Bigr),
\label{eq:E-def}
\end{equation}
with \(d_R=\dim R\). When a real adjoint basis is chosen and the kernel is real, one can replace $\mathcal{N}^\dagger$ by $\mathcal{N}^{\rm T}$; we keep the Hermitian form for generality.

A useful singular-value form is as follows. Let $s_i$ be the singular values of $\mathcal{N}$ (eigenvalues of $(\mathcal{N}\mathcal{N}^\dagger)^{1/2}$).
Then $\rho_R$ has spectrum $\{\lambda_i\}$ with
\(
\lambda_i=s_i^2\big/\sum_j s_j^2
\).
Therefore
\begin{equation}
E=\frac{d_R}{d_R-1}\left(1-\frac{\sum_i s_i^4}{\left(\sum_i s_i^2\right)^2}\right).
\label{eq:E-sv}
\end{equation}
This makes clear that: (i) $E$ is invariant under an overall rescaling of $\mathcal{N}$; (ii) $E=0$ iff only one $s_i$ is nonzero (a product state); and (iii) if exactly $k$ singular values are equal and nonzero, $E=\dfrac{d_R}{d_R-1}\left(1-\dfrac{1}{k}\right)$. The singular values of \(\mathcal{N}\) thus determine $E$;
$E=0$ iff only one singular value is nonzero (a separable output),
and $E$ reaches its maximum $E_{\max}=1$ for a maximally mixed reduced state.

Whenever convenient we will specialize to \(R=R'=\mathrm{Adj}\), which corresponds to $2\!\to\!2$ scattering of adjoint states $g$ of SU(N):  $g_a(p_1)+g_b(p_2)\to g_c(p_3)+g_d(p_4)$.  For SU(3), $g$ are the usual gluons, and we will use this terminology for other gauge groups (for example SU(2)) as well. A two-gluon color state lives in ${\rm Adj}\!\otimes\!{\rm Adj}$ with basis $|a,b\rangle:=|a\rangle\otimes|b\rangle$. Then we will work in the adjoint-index basis $\{|a\rangle\}$ of $\mathrm{SU}(N)$, with $a=1,\dots,N_A$ and $N_A=\dim(\mathrm{Adj})=N^2-1$. Until Sec.~\ref{YMkins2}, however, we will not introduce or use any kinematic relations, and work only with the color structure in Eq.~\eqref{eq:defK}. Thus, we will be agnostic about the dependence of $K$ on the set of variables $\Theta$ (and often keep the dependence on $\Theta$ implicit). In particular, we will not utilize color-kinematics duality.

\subsection{Representation Content and Entanglement Structure}
\label{sec:YM-basics}

The degree to which \(K\) entangles or preserves separability
is dictated purely by the decomposition of \(R\!\otimes\!R'\). We will illustrate this in the context of Yang-Mills theory, explicitly working out the cases of SU(2) and SU(3), taking two benchmark cases for the representation $R$ in each case: the adjoint and the fundamental.

In Yang-Mills theory the color-dressed amplitude at four points
is a linear combination of the three “double-$f$” tensors,
\begin{equation}
c_s=f_{abe}f_{cde},\qquad
c_t=-f_{ace}f_{bde},\qquad
c_u=f_{ade}f_{bce},\qquad
c_s+c_t+c_u=0,
\label{eq:doublef-Jacobi}
\end{equation}
so that
\begin{equation}
K(\Theta)\ \propto\ A_s(\Theta)\,c_s\;+\;A_t(\Theta)\,c_t\;+\;A_u(\Theta)\,c_u,
\qquad \text{with }A_i(\Theta)\ \text{kinematic weights}.
\label{eq:K-from-ff}
\end{equation}

We will be interested in the basis of SU(N)-equivariant endomorphisms of \(R\!\otimes\!R'\). Here, “equivariant” means endomorphisms on \(R\,\!\otimes\!\,R'\) that commute with the diagonal SU(N) action (i.e. are invariant under conjugation $T\mapsto U T U^{-1}$), whereas an “invariant tensor” means a group-fixed vector inside a tensor product space. In practice we use the invariant tensors $(\delta_{ab},f_{abc},d_{abc})$ to build the equivariant operators that act on \(R\,\!\otimes\!\,R'\). We now explicitly show the case of SU(2), and then discuss SU(3) in detail.

\subsection{SU(2): Kernel on ${\rm Adj}\!\otimes\!{\rm Adj}$ }

For $SU(2)$, $d_{abc}=0$ and the  set $\{\mathbb I, \mathbb S,P_{\text{singlet}}\}$ defined below already spans all equivariant endomorphisms on $\mathbf 3\otimes\mathbf 3$. In terms of the invariant tensors, these are given by
\begin{align}
(\mathbb I)_{cd,\,ab} &:= \delta_{ac}\,\delta_{bd} , && \text{(identity on the pair)} \label{eq:Iop}\\
(\mathbb S)_{cd,\,ab} &:= \delta_{ad}\,\delta_{bc} , && \text{(swap: $\mathbb S\,|a,b\rangle=|b,a\rangle$)} \label{eq:Sop}\\
(P_{\rm singlet})_{cd,\,ab} &:= \frac{1}{N_A}\,\delta_{ab}\,\delta_{cd} , && \text{(projector onto the singlet in ${\rm Adj}\!\otimes\!{\rm Adj}$)} \label{eq:Psing}
\end{align}
with $P_{\rm singlet}^2=P_{\rm singlet}$ and $\mathbb S^2=\mathbb I$. Any SU($2$)-invariant rank-$4$ tensor built from 
$\delta_{ab}$ and $f_{abc}$ can be reduced to a linear combination of the three operators Eq.~\eqref{eq:Iop}–Eq.~\eqref{eq:Psing}. Consequently, for some scalar functions $a(\Theta),b(\Theta),c(\Theta)$, we can write Eq.~(\ref{eq:K-from-ff}) as
\begin{equation}
K(\Theta)\;\propto\; a(\Theta)\,\mathbb I \;+\; b(\Theta)\,\mathbb S \;+\; c(\Theta)\,P_{\rm singlet}\,,
\label{eq:K-ISP}
\end{equation}
where the overall proportionality is irrelevant for entanglement once the out-state is normalized.

One can relate  $\{c_s,c_t,c_u\}$ to $\{\mathbb I,\mathbb S,P_{\rm singlet}\}$ by comparing Eq.~(\ref{eq:K-from-ff}) to Eq.~(\ref{eq:K-ISP}). For $\mathrm{SU}(2)$ one has $f_{abc}=\varepsilon_{abc}$ and $d_{abc}=0$, and the identity
\begin{equation}
\varepsilon_{ace}\,\varepsilon_{bde}
\;=\;
\delta_{ab}\delta_{cd}\;-\;\delta_{ad}\delta_{bc}
\;=\; (N_A\,P_{\rm singlet}-\mathbb S)_{cd,\,ab}.
\label{eq:eps-eps}
\end{equation}
By permuting the indices in Eq.~\eqref{eq:eps-eps} one obtains
\begin{equation}
\begin{aligned}
c_s &= f_{abe}f_{cde} \;=\; \mathbb I - \mathbb S,\\
c_t &= -f_{ace}f_{bde} \;=\; -N_A\,P_{\rm singlet} + \mathbb S,\qquad (N_A=3),\\
c_u &= f_{ade}f_{bce} \;=\; N_A\,P_{\rm singlet} - \mathbb I.
\end{aligned}
\label{eq:csctcu-in-ISP}
\end{equation}
Therefore, with $A_i(\Theta)$ the kinematic weights multiplying each color tensor in the amplitude,
\begin{align}
K(\Theta)
&\propto A_s\,c_s + A_t\,c_t + A_u\,c_u \nonumber\\
&= A_s(\mathbb I-\mathbb S) + A_t(-N_A P_{\rm singlet}+\mathbb S) + A_u(N_A P_{\rm singlet}-\mathbb I) \nonumber\\[2pt]
&= \underbrace{\big(A_s - A_u\big)}_{=:a(\Theta)}\,\mathbb I
\;+\; \underbrace{\big(-A_s + A_t\big)}_{=:b(\Theta)}\,\mathbb S
\;+\; \underbrace{N_A\big(-A_t + A_u\big)}_{=:c(\Theta)}\,P_{\rm singlet}.
\label{eq:Kcoeffs-SU2}
\end{align}
Equations Eq.~\eqref{eq:K-ISP} and Eq.~\eqref{eq:Kcoeffs-SU2} give an explicit map between the delta-tensor expressions and the operator basis $\{\mathbb I,\mathbb S,P_{\rm singlet}\}$.

\subsection{SU(3): Kernel on ${\rm Adj}\!\otimes\!{\rm Adj}$ }

For  SU(3), we  will intentionally restrict to the block-diagonal invariants
\[
\mathbb I,\qquad \mathbb S,\qquad P_{\text{singlet}},\qquad P_{8_S},\qquad P_{8_A},
\]
which are the identity, the swap $\mathbb S$, the singlet projector, and the two octet projectors (symmetric and antisymmetric) on $\mathbf 8\otimes\mathbf 8$. From these five one obtains
\[
P_{27}=\tfrac12\,(\mathbb I + \mathbb S)-P_{\text{singlet}}-P_{8_S},\qquad
P_{10\oplus\overline{10}}=\tfrac12\,(\mathbb I- \mathbb S)-P_{8_A},
\]
so this basis is sufficient to express the YM color kernel built from the double-$f$ structures relevant for YM\footnote{We note that this is not the full equivariant algebra. The complete algebra of $SU(3)$-equivariant endomorphisms on $\mathbf 8\otimes\mathbf 8$ is
\[
\mathrm{End}_{SU(3)}(\mathbf 8\otimes\mathbf 8)\;\cong\;
\mathbb C\oplus\mathbb C\oplus\mathbb C\oplus\mathbb C\;\oplus\;M_2(\mathbb C),
\]
i.e. scalars on $\mathbf 1,\mathbf{27},\mathbf{10},\overline{\mathbf{10}}$ together with an arbitrary $2\times2$ action on the two octets $\mathbf 8_S,\mathbf 8_A$. To span the off-diagonal octet intertwiners (mapping $\mathbf 8_S\leftrightarrow \mathbf 8_A$) one needs $df$-type operators, which anticommute with $\mathbb S$ and therefore cannot be generated by polynomials in $\{\mathbb I,\mathbb S,P_{\text{singlet}},P_{8_S},P_{8_A}\}$ (those all commute with $\mathbb S$ and are block-diagonal). However, we will not be interested in $df$-type operators in the context of YM.}.

For $\mathrm{SU}(N\!\ge\!3)$ the totally symmetric tensor $d^{abc}$ is nonzero, and the
double-$f$ color structures appearing at four points reduce to a sum of
$\delta\delta$ terms plus $d\,d$ terms. It is convenient to extend the operator
basis we have introduced thus far to include two independent $d\,d$ invariants.

We now introduce the operator basis and $d\,d$ tensors. Recall
\begin{align}
(\mathbb I)_{cd,\,ab} := \delta_{ac}\,\delta_{bd},\qquad
(\mathbb S)_{cd,\,ab} := \delta_{ad}\,\delta_{bc},\qquad
(P_{\rm singlet})_{cd,\,ab} := \frac{1}{N_A}\,\delta_{ab}\,\delta_{cd},
\end{align}
with $N_A=\dim(\mathrm{Adj})=N^2-1$.
When $d^{abc}\!\neq\!0$ we also define the three rank-$4$ tensors
\begin{equation}
(\mathbb D_s)_{cd,\,ab} := d_{abe}\,d_{cde},\qquad
(\mathbb D_t)_{cd,\,ab} := d_{ace}\,d_{bde},\qquad
(\mathbb D_u)_{cd,\,ab} := d_{ade}\,d_{bce}.
\label{eq:Dstu-def}
\end{equation}
They are not all independent:  for SU(3), there is one linear relation among $\{\mathbb D_s,\mathbb D_t,\mathbb D_u\}$,
so the $d\,d$ sector is effectively two–dimensional:
\begin{equation}
 \mathbb D_s + \mathbb D_t + \mathbb D_u \, = \, \frac{1}{3} (N_A P_{\rm singlet} + \mathbb I + \mathbb S)
\end{equation}
A convenient independent pair is
\((\mathbb D_t-\mathbb D_u,\ \mathbb D_u-\mathbb D_s)\). The independent operator set we choose is thus $\{\mathbb I, \mathbb S, P_{\rm singlet}, \mathbb D_t-\mathbb D_u, \mathbb D_u-\mathbb D_s\}$. Equivalently, one could have chosen the set  $\{\mathbb I, \mathbb S, P_{\rm singlet}, P_{8_S},P_{8_A}\}$.

Using the standard identity (with our adjoint-index conventions)~\cite{Haber:2019sgz}
\begin{equation}
f_{abe}f_{cde}
= \frac{2}{N}\,\big(\delta_{ac}\delta_{bd}-\delta_{ad}\delta_{bc}\big)
 \;+\; \big(d_{ace}d_{bde}-d_{ade}d_{bce}\big),
\label{eq:ff-reduction}
\end{equation}
and its index permutations, the three channel color factors become

\begin{equation}
\begin{aligned}
c_s &= f_{abe}f_{cde}
= \frac{2}{N}\,\big(\mathbb I-\mathbb S\big) \;+\; \big(\mathbb D_t-\mathbb D_u\big),\\[2pt]
c_t &= -f_{ace}f_{bde}
= \frac{2}{N}\,\big(-N_A P_{\rm singlet}+\mathbb S\big) \;+\; \big(-\mathbb D_s+\mathbb D_u\big),\\[2pt]
c_u &= f_{ade}f_{bce}
= \frac{2}{N}\,\big( N_A P_{\rm singlet} - \mathbb I\big) \;+\; \big(\mathbb D_s-\mathbb D_t\big).
\end{aligned}
\label{eq:csctcu-SUN}
\end{equation}

The explicit coefficients in the kernel can now be computed. The color kernel
\(
K(\Theta)\propto A_s c_s + A_t c_t + A_u c_u
\)
may be written in a minimal invariant basis as
\begin{align}
K(\Theta)
&\propto a(\Theta)\,\mathbb I \;+\; b(\Theta)\,\mathbb S \;+\; c(\Theta)\,P_{\rm singlet}
\;+\; d_1(\Theta)\,\big(\mathbb D_t-\mathbb D_u\big)
\;+\; d_2(\Theta)\,\big(\mathbb D_u-\mathbb D_s\big),
\label{eq:K-SU3-basis}
\end{align}
with coefficients (specializing Eq.~\eqref{eq:csctcu-SUN} to $N{=}3$ and $N_A{=}8$)
\begin{equation}
\begin{aligned}
a(\Theta) &= \frac{2}{N}\,\big(A_s - A_u\big)\Big|_{N=3} \;=\; \frac{2}{3}\,\big(A_s - A_u\big),\\[4pt]
b(\Theta) &= \frac{2}{N}\,\big(-A_s + A_t\big)\Big|_{N=3} \;=\; \frac{2}{3}\,\big(-A_s + A_t\big),\\[4pt]
c(\Theta) &= \frac{2}{N}\,N_A\,\big(-A_t + A_u\big)\Big|_{N=3,N_A=8} \;=\; \frac{16}{3}\,\big(-A_t + A_u\big),\\[4pt]
d_1(\Theta) &= \ A_s - A_u,\\[2pt]
d_2(\Theta) &= \ A_t - A_u.
\end{aligned}
\label{eq:coeffs-SU3}
\end{equation}
(Any overall proportionality in Eq.~\eqref{eq:K-SU3-basis} drops out once the out–state is normalized for entanglement.)

A few limits and special kinematics can be checked.

\begin{itemize}
\item \emph{SU(2) limit.} If $d^{abc}=0$ (as in SU(2)), then $\mathbb D_{s,t,u}=0$ and Eq.~\eqref{eq:K-SU3-basis} reduces to the SU(2) formula: only $(\mathbb I,\mathbb S,P_{\rm singlet})$ appear, with
$\ a=\tfrac{2}{N}(A_s-A_u)$, $b=\tfrac{2}{N}(-A_s+A_t)$, $c=\tfrac{2}{N}N_A(-A_t+A_u)$ evaluated at $N=2$, $N_A=3$.

\item \emph{Jacobi sum.} The combination $c_s+c_t+c_u=0$ holds identically by the Lie-algebra Jacobi identity (independently of $N$); in the basis Eq.~\eqref{eq:csctcu-SUN} this is reflected by the cancellations
\(
\tfrac{2}{N}\big[(N_A P_{\rm singlet}-\mathbb S)+(\mathbb I-N_A P_{\rm singlet})+(\mathbb S-\mathbb I)\big]=0
\)
and
\(
(\mathbb D_s-\mathbb D_u)+(\mathbb D_t-\mathbb D_s)+(\mathbb D_u-\mathbb D_t)=0.
\)
\end{itemize}

\subsection{Representation Structure of Scattering States and Entanglement: Fundamentals versus Adjoints}

As we have stressed, the most important quantity determining the entanglement is the \emph{representation structure}
of the internal Hilbert space of the incoming particles.
Once the relevant tensor product decomposition is fixed,
the degree to which the scattering operator $K$ entangles or preserves separability
is determined purely by the representation content, independently of whether the
underlying symmetry is realized as a global flavor group or as a local gauge group.

It is useful to restate what we have discussed thus far in a slightly more general language. Let external lines carry representations $R$ and $R'$ of SU(N), so the initial Hilbert space is 
$\mathcal H_R\otimes\mathcal H_{R'}$ with basis $|i\rangle\otimes|j\rangle$.
The scattering matrix induces a linear map
\[
K:\ |i,j\rangle\xrightarrow{K_{k\ell,\,ij}}\,|k,\ell\rangle,\qquad
K_{k\ell,\,ij}\ \in\ \mathrm{End}^{\rm SU(N)}(\mathcal H_R\otimes\mathcal H_{R'})\,.
\]
$K$ lives in the commutant algebra generated by the projectors $P_{\mathcal R}$ onto irreps ${\mathcal R}\subset R\otimes R'$.
Equivalently, in a permutation‐style operator basis one keeps only SU(N)-invariant rank-4 tensors built from $\delta$, $f$, $d$ appropriate to $R,R'$. We now turn to specific cases.

\subsubsection{Fundamental-Fundamental Scattering}\label{funfund}

\textit{Fundamental $\otimes$ Fundamental:} For $R=R'=\mathbf N$ one has the decomposition $\mathbf N\otimes\mathbf N=\mathrm{Sym}^2\mathbf N\ \oplus\ \wedge^2\mathbf N$.
For $N\ge3$ there is no singlet in $\mathbf N\otimes\mathbf N$.
Consequently the SU(N)-invariant operator algebra on $\mathbf N\otimes\mathbf N$ is two‐dimensional, spanned by
\[
\mathbb I_{k\ell,\,ij}=\delta_{ik}\delta_{j\ell},\qquad 
\mathbb S_{k\ell,\,ij}=\delta_{i\ell}\delta_{jk}.
\]
Hence the most general invariant kernel is
\[
K\ =\ a\,\mathbb I\ +\ b\,\mathbb S\,,\qquad (N\ge3).
\]
For SU(2), $\mathbf2\otimes\mathbf2=\mathbf3\oplus\mathbf1$ and a singlet does exist; one may equivalently work with projectors onto triplet and singlet or with $\{\mathbb I,\mathbb S\}$ plus the antisymmetric invariant $\epsilon_{ij}$ (fermion statistics may already single out the antisymmetric channel).

The example of SU(3)-equivariant endomorphisms of $\mathbf 3 \otimes \mathbf 3$ can be worked out. The Clebsch--Gordan decomposition
\[
\mathbf 3\otimes \mathbf 3 \;\cong\; \mathbf{6}\;\oplus\;\overline{\mathbf{3}}
\]
has each irrep with multiplicity $1$. By Schur's lemma,
\[
\mathrm{End}_{SU(3)}(\mathbf 3\otimes \mathbf 3)\;\cong\;\mathbb{C}\oplus\mathbb{C},
\]
i.e. an SU(3)-equivariant endomorphism acts as an independent scalar on the $\mathbf{6}$ and on the $\overline{\mathbf{3}}$. The symmetric/antisymmetric projectors are
\[
P_{\mathbf{6}}=\tfrac12\,(\mathbb I+ \mathbb S),\qquad
P_{\overline{\mathbf{3}}}=\tfrac12\,(\mathbb I- \mathbb S),\qquad
P_{\mathbf{6}}+P_{\overline{\mathbf{3}}}=\mathbb I,
\]
and
\[
\mathrm{End}_{SU(3)}(\mathbf 3\otimes \mathbf 3)=\mathrm{Span}\{P_{\mathbf{6}},\,P_{\overline{\mathbf{3}}}\}
=\mathrm{Span}\{\mathbb I,\, \mathbb S\}.
\]
There is no singlet in $\mathbf 3\otimes \mathbf 3$, hence no $P_{\rm singlet}$ term; and since $\mathbf{6}$ and $\overline{\mathbf{3}}$ appear with multiplicity $1$, no equivariant mixing between them is allowed.

No additional invariant tensors exist. Any theory in which the scattering particles transform in the fundamental
representation - whether a global $SU(N)$ flavor theory or a gauge theory such as QCD
with quarks - therefore produces, at the group-theoretic level, a two-parameter minimally entangling operator of the form $\mathrm{Span}\{\mathbb I,\, \mathbb S\}$. In this sense, quark-quark scattering and any analogous fundamental-fundamental process is ``minimally entangling''.
While the two-dimensional space $\mathrm{Span}\{\mathbb{I},\mathbb{S}\}$ usually generates entanglement for generic
complex coefficients $a,b$ (the physical content of the scattering, which includes the energy or angle dependence of the coefficients
$a,b$, resides entirely within this two–dimensional operator space), the  process is ``minimally entangling'' in the sense that
no further independent invariant tensors are available to enhance the entangling power. To elaborate further, for $K(\theta)\;=\;a(\theta)\,\mathbb{I} + b(\theta)\,\mathbb{S}$ and input product state $|u\rangle\!\otimes\!|v\rangle$,
the (unnormalized) output is
\[
|{\rm out}\rangle = a(\theta)\,|u\rangle\!\otimes\!|v\rangle
                  + b(\theta)\,|v\rangle\!\otimes\!|u\rangle.
\]
When $b(\theta)$ or $a(\theta)$ is vanishing,
the operator reduces to $\mathbb{I}$ or $\mathbb{S}$, both of which map separable states to separable states.
For any other ratio or phase difference between $a$ and $b$, a generic product
input acquires nonzero Schmidt rank:
the two terms above are not collinear in $\mathcal{H}_A\otimes\mathcal{H}_B$.
Hence, within the two-dimensional subspace
$\mathrm{Span}\{\mathbb{I},\mathbb{S}\}$,
only the directions $\mathbb{I}$ and $\mathbb{S}$ themselves are separability preserving; all other points in the plane generate entanglement.

We now go on to study the entanglement of fundamental scattering in Yang-Mills in more detail, explicitly working out the maximum entanglement.
Since the kernel $K(\theta)$ is invariant under the SU(N) transformation, one can always choose the rotation so that $u = (1,0,\cdots)^T$ and $v =(\cos\vartheta,\sin\vartheta \, e^{i \eta},\cdots)^T$. Here $\cos\vartheta = |u^\dagger v|$ is invariant under the rotation and $\eta$ is the possible additional phase. We use the rescaled kernel function 
\[
K \sim \mathbb{I} + z \, \mathbb{S},\qquad
z \, \equiv \, b(\theta)/a(\theta)\,\,.
\]
 Directly computing the final state and calculating its entanglement power, one can obtain
\begin{equation}\label{eq:fund_ent_formula}
    E(|\text{out}\rangle) = \dfrac{2N \, z^2 \sin^4 \vartheta}{(N-1) \brr{1+z^2 + 2 z \cos^2\vartheta}^2 } \ .
\end{equation}
It is easily checked that this value is invariant under the inversion transformation $z \rightarrow 1/z$. A flat direction occurs at $z = -1$, in which the entanglement is equal to $N/(2(N-1))$ independent of $\vartheta$. 
\begin{figure}
    \centering
    \includegraphics[width=0.6\linewidth]{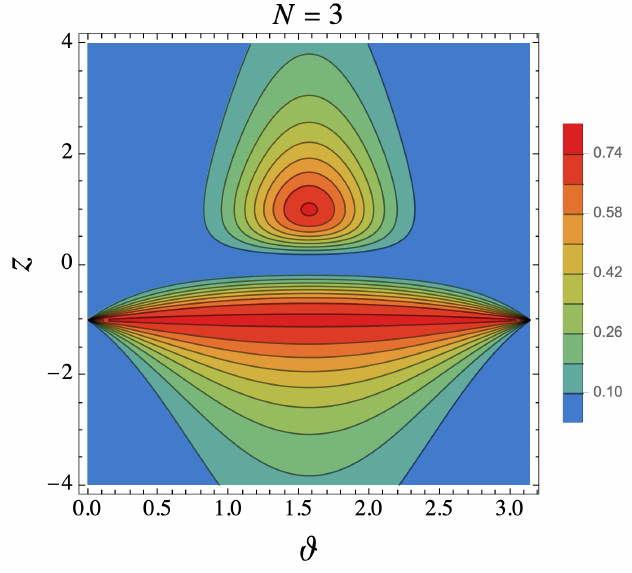}
    \caption{Entanglement power as a function of $z$ and $\vartheta$ as in Eq.(\ref{eq:fund_ent_formula}) for the case of  Fundamental-Fundamental scattering. Here we take $N = 3$ as the benchmark. }
    \label{fig:fund_ent}
\end{figure}
For fixed $N$ and other $z$ values, one can show that the maximal entanglement value is achieved at $\vartheta = \pi/2$, equaling to 
\begin{equation}
  E_{\rm max}\brr{\vartheta = \dfrac{\pi}{2}} =  \dfrac{2 N z^2}{(N-1)(1+z^2)^2} = \dfrac{2N}{(N-1)\brr{z+\frac{1}{z}}^2} \leq \dfrac{N}{2(N-1)} \ .
\end{equation}
In the last term, the equality only holds at $z = 1$. 
We display the entanglement power on the $(z,\vartheta)$ plane as in Fig.~\ref{fig:fund_ent}. The red region refers to the high entanglement case, concentrating around $z= -1$ and $(z = 1, \vartheta=\pi/2)$.
Therefore one can reproduce the maximal entanglement to be 1 for the case of SU(2). For $N = 3$, this maximal value equals $3/4$. As $N$ approaches infinity, the maximal entanglement tends to the fixed value $1/2$.

\subsubsection{Fundamental $\otimes$ Anti‐fundamental}
For $R=\mathbf N$, $R'=\overline{\mathbf N}$ one has $\mathbf N\otimes\overline{\mathbf N}=\mathbf1\oplus\mathrm{Adj}$.
The SU($N$)-invariant operator algebra is again two‐dimensional, now spanned by
\[
\mathbb I_{k\bar\ell,\,i\bar j}=\delta_{ik}\,\delta_{\bar j \bar \ell},\qquad 
P_{\rm{singlet},\,k\bar\ell,\,i\bar j}=\tfrac{1}{N}\,\delta_{i\bar j}\,\delta_{k\bar\ell}\,.
\]
(Note $\mathbb S$ does not act within the same space here.) Thus
\[
 K\ =\ \tilde a\,\mathbb I\ +\ \tilde c\,P_{\rm{singlet}}\,.
\]
Like adjoint-adjoint scattering, this case is also generally entangling, due to the singlet projection.

\subsubsection{Adjoint-Adjoint Scattering}

\textit{Adjoint $\otimes$ Adjoint:} For $R=R'=\mathrm{Adj}$ the minimal invariant basis at 4-pt is
\[
\big\{\ \mathbb I_{cd,\,ab}=\delta_{ac}\delta_{bd},\quad 
\mathbb S_{cd,\,ab}=\delta_{ad}\delta_{bc},\quad
P_{\rm{singlet},\,cd,\,ab}=\tfrac{1}{N_A}\delta_{ab}\delta_{cd}\ \big\}
\]
plus, when $N\ge3$, two independent $dd$-type tensors (e.g. $\mathbb D_t-\mathbb D_u$, $\mathbb D_u-\mathbb D_s$ with $\mathbb D_s{:=}d_{abe}d_{cde}$, etc.).
Thus
\[
K \;\propto\; a\,\mathbb I+b\,\mathbb S+c\,P_{\rm singlet}\;+\;\begin{cases}
0,& N=2,\\
d_1(\mathbb D_t-\mathbb D_u)+d_2(\mathbb D_u-\mathbb D_s),& N\ge3,
\end{cases}
\]
with coefficients fixed by the kinematic weights.

Unlike in the fundamental case,
there is no way to switch off the mixing between these
$P_{\text{singlet}}$ and $dd$ sectors in general.
Generic adjoint-adjoint scattering therefore inevitably maps a product (minimally entangled) initial state to an entangled final state.
Equivalently, the operator $K$ acting on $\mathrm{Adj}\otimes\mathrm{Adj}$
has nontrivial components outside $\mathrm{Span}\{\mathbb{I},\mathbb{S}\}$,
so adjoint scattering is intrinsically an \emph{entangler}.
This statement depends only on the representation content,
not on whether the adjoint particles are gauge bosons or
transform under some global adjoint symmetry.

\subsection{Adjoint Scattering is an Entangler}\label{entangdiagonal}

The contrast between the  case of fundamental versus adjoint scattering highlights that it is
the \emph{representation of the scattering states} and not the rank of the group,
nor the distinction between global and gauge realizations, that controls
the minimal entanglement structure of the $S$-matrix.
A fundamental scattering with a fundamental is group-theoretically
confined to a two–dimensional invariant space and is thus minimally entangling along the rays $\mathbb{I}$ and $\mathbb{S}$; an adjoint scattering with an adjoint 
 has a much larger invariant algebra and is
generically entangling.
The particular coefficients of the invariant tensors (and their
kinematic dependence) encode the detailed dynamics, but the
entanglement behavior follows directly from representation theory. For special kinematic configurations or fine-tuned coupling relations,
some of the entangling channels in the adjoint case may vanish,
yielding reduced entanglement, but these are exceptions rather than generic.

We now discuss, at a slightly more detailed level, the nature of adjoint-adjoint scattering and its relation to minimal entanglement. “Minimal entanglement” would require that the map $K$ be separability-preserving for all product inputs in $\mathrm{Adj}\otimes\mathrm{Adj}$ (up to local adjoint rotations on each leg). On ${\rm Adj}\otimes{\rm Adj}$ the SU(N)-invariant operator algebra at 4-pt includes
$\{\mathbb I,\mathbb S,P_{\rm singlet}\}$ and, for $N\!\ge\!3$, two independent $d\,d$ tensors.
Requiring $K$ to be minimally entangling (separability preserving) projects out $P_{\rm singlet}$ and the $d\,d$ sector, leaving
$K\in\mathrm{span}\{\mathbb I,\mathbb S\}$ up to local adjoint rotations. Allowing a $P_{\rm singlet}$ or $d\,d$ admixture generally generates  entanglement.

For example, one can check the case of $P_{\rm singlet}$. For a product basis vector $|a\rangle\otimes|b\rangle$,
\[
P_{\rm singlet}:\ |a\rangle\otimes|b\rangle\ \longmapsto\ \frac{\delta_{ab}}{N_A}\sum_{c}|c\rangle\otimes|c\rangle .
\]
Except on the measure-zero set with $\delta_{ab}=0$, the image is a fixed Bell-type vector of Schmidt rank $N_A>1$, hence entangled. Therefore any non-zero coefficient in front of  $P_{\rm singlet}$ in the decomposition of $K$ violates separability preservation for some product inputs.

A similar conclusion holds for the $\mathbb D$ operators. The action of $\mathbb D_u$, for example, on a product vector is
\[
\mathbb D_u:\ |a\rangle\otimes|b\rangle\ \longmapsto\ \sum_{c,d}\Big(\sum_{e}d_{abe}\,d_{cde}\Big)\,|c\rangle\otimes|d\rangle
\;=\;\sum_{e}d_{abe}\,|\psi_e\rangle,
\quad |\psi_e\rangle:=\sum_{c,d}d_{cde}\,|c\rangle\otimes|d\rangle .
\]
For generic $(a,b)$ the coefficients $d_{abe}$ do not all vanish, and each $|\psi_e\rangle$ has Schmidt rank $>1$ (the symmetric $d$ correlates the two legs). Hence $\mathbb D_u$ (and any nontrivial linear combination of $\mathbb D_{s,t,u}$) produces entangled outputs for an open set of product inputs. The same conclusion holds for $(\mathbb D_t-\mathbb D_u)$ and $(\mathbb D_u-\mathbb D_s)$ that span the $dd$ subspace.

On the other hand, one can check explicitly that the swap operator preserves separability, by computing the adjoint-space commutator. Let the color Hilbert space for each gluon be \(\mathcal H_{\rm Adj}\cong\mathbb C^{N_A}\) (with \(N_A=N^2-1\)).
Define the swap operator \(\mathbb S:\mathcal H_{\rm Adj}\otimes\mathcal H_{\rm Adj}\to\mathcal H_{\rm Adj}\otimes\mathcal H_{\rm Adj}\) by
\[
\mathbb S\,|a,b\rangle \;=\; |b,a\rangle,
\qquad\text{i.e.}\qquad
\mathbb S_{k\ell,\,ij}=\delta_{i\ell}\,\delta_{jk}.
\]
Let $SU(N)_A$  and $SU(N)_B$ be independent transformations on the two initial states.  Let \(Q_A\) and \(Q_B\) be (Hermitian) generators acting on the first and second factor, respectively; in the adjoint representation \(Q=\theta^A T_{A,\rm Adj}\) with \((T_{A,\rm Adj})_{ab}=-i f^{Aab}\) (any faithful adjoint-rep choice works). Consider the additive charge
\[
\mathcal Q \;=\; Q_A\otimes \mathbf 1 \;+\; \mathbf 1\otimes Q_B.
\]
$\mathbb S$ satisfies the identities, for any linear operator \(X\) on \(\mathcal H_{\rm Adj}\),
\begin{equation}
\mathbb S\,(X\otimes \mathbf 1) \;=\; (\mathbf 1\otimes X)\,\mathbb S,
\qquad
\mathbb S\,(\mathbf 1\otimes X) \;=\; (X\otimes \mathbf 1)\,\mathbb S,
\label{eq:swap-idents}
\end{equation}
and \(\mathbb S^2=\mathbf 1\).
Then
\begin{align}
[\mathbb S,\mathcal Q]
&= \mathbb S(Q_A\otimes \mathbf 1) - (Q_A\otimes \mathbf 1)\mathbb S
   \;+\; \mathbb S(\mathbf 1\otimes Q_B) - (\mathbf 1\otimes Q_B)\mathbb S \nonumber\\
&= (\mathbf 1\otimes Q_A)\mathbb S - (Q_A\otimes \mathbf 1)\mathbb S
   \;+\; (Q_B\otimes \mathbf 1)\mathbb S - (\mathbf 1\otimes Q_B)\mathbb S \nonumber\\
&=\big[(Q_B-Q_A)\otimes \mathbf 1 \;-\; \mathbf 1\otimes (Q_B-Q_A)\big]\,\mathbb S.
\label{eq:comm-result}
\end{align}
Since \(\mathbb S\) is invertible (\(\mathbb S^{-1}=\mathbb S\)), the commutator vanishes iff
\(
(Q_B-Q_A)\otimes \mathbf 1 \;=\; \mathbf 1\otimes (Q_B-Q_A).
\)
This implies \(Q_B-Q_A=c\,\mathbf 1\) on \(\mathcal H_{\rm Adj}\); but adjoint generators are traceless, hence \(c=0\).
Therefore
\[
[\mathbb S,\mathcal Q]=0 \quad\Longleftrightarrow\quad Q_A=Q_B,
\]
i.e. only the diagonal \(SU(N)_D\subset SU(N)_A\times SU(N)_B\) commutes with \(\mathbb S\).

\section{Yang Mills:  Color-Kinematics Duality and the Scattering Kernel }\label{YMkins2}

Up to now, we have only exploited the double-$f$ structure of the Yang-Mills amplitude. At this stage, further simplification of Eq.~(\ref{eq:coeffs-SU3}) is possible by assuming the standard color-dressed Yang-Mills amplitude form
\begin{eqnarray}
  K\,\,\equiv\,\,M_4^{\text{full,tree}} &=& \frac{n_sc_s}{s} +  \frac{n_tc_t}{t} + \frac{n_uc_u}{u} \nonumber \\
&=&  A_s\,c_s+\;A_t\,c_t+A_u\,c_u   \ ,
\end{eqnarray}
Then, one has $A_i(\Theta)=n_i/s_i$ with $s_i\in\{s,t,u\}$. The numerator $n_i$ factors depend on the kinematics and the external polarization vectors. 
It should be noted that each $A_i$ by itself is not gauge invariant;  only the whole amplitude is. From color-kinematics duality (CKD)~\cite{Bern:2008qj}, one  can always find a gauge choice such that 
\begin{equation}\label{eq:kin_jacobi1}
   s A_s + t A_t + u A_u = 0 \ .
\end{equation}
More details about amplitudes in pure Yang-Mills is provided in Appendix~\ref{ymscattbasics}. Using Eq.~(\ref{eq:kin_jacobi1}) along with $s+t+u=0$, one can simply get
\begin{equation}
    K = (A_s - A_t) \dfrac{c_s \, t - c_t\, s}{s+t}
\end{equation}
The full amplitude can be expressed as the product of the group structure part and the (helicity dependent) kinematic part. More explicitly, the variables $\{a,b,c,d_1,d_2\}$ in Eq.~(\ref{eq:coeffs-SU3}) can all be expressed in terms of the variable $a$ only, as follows
\begin{equation}
\begin{aligned}
b &= \left(\frac{u}{t}\right)\,a ,\\[2pt]
c &= N_A \,\left(\frac{s}{t}\right)\,a,\\[2pt]
d_1 &= \frac{N}{2}\,a,\\[2pt]
d_2 &= -\frac{N}{2}\left(\frac{s}{t}\right)\,a,
\end{aligned}
\label{eq:afterkcdual}
\end{equation}
so that one obtains 
\begin{align}
K(\Theta)
&\propto a(\Theta)\,\left[\mathbb I \;+\; \left(\frac{u}{t}\right)\,\mathbb S \;+\; N_A \left(\frac{s}{t}\right)\,P_{\rm singlet}
\;+\;\frac{N}{2}\,\big(\mathbb D_t-\mathbb D_u\big)
\;-\; \frac{N}{2}\left(\frac{s}{t}\right)\,\big(\mathbb D_u-\mathbb D_s\big)\right],
\label{eq:K-SU3-basis-after-KC}
\end{align}

\subsection{CKD guarantees separability of color and helicity Hilbert spaces}\label{subsec:scope-doublef}

Up to the end of  Sec.~\ref{entangdiagonal} we have made no use of CKD and remained agnostic about the kinematic functions $A_i(\Theta)$. Our only structural input is that the color kernel be a linear combination of the three double-$f$ tensors. With this assumption alone, all results up to Sec.~\ref{entangdiagonal} follow:
\begin{enumerate}
\item \textit{SU(2) and SU(3) operator expansions.} The reductions of $c_{s,t,u}$ onto the SU($N$)-equivariant operator basis
\(
\{\mathbb I,\mathbb S,P_{\mathbf1}\}\ (N=2)
\)
and
\(
\{\mathbb I,\mathbb S,P_{\mathbf1},\mathbb{D}_t-\mathbb{D}_u,\mathbb{D}_u-\mathbb{D}_s\}\ (N=3)
\)
are purely algebraic and hold for any $A_i(\Theta)$.
\item \textit{Intrinsic entanglers versus minimal entanglers:} The separability-preserving (``minimal entangler'') requirement projects $K$ onto $\mathrm{span}\{\mathbb I,\mathbb S\}$ and implies that only the diagonal subgroup $SU(N)_D$ commutes with the kernel; this is independent of any relation among $A_i(\Theta)$.
\end{enumerate}
In contrast, Yang-Mills-specific input first appears in Sec.~\ref{YMkins2}, where CKD appears, allowing the operator coefficients to be expressed in terms of a single function and kinematic ratios. 

The kinematics data is determined by the scattering angles and the helicities of the external particles. For the scattering of gauge bosons, then, there are two relevant Hilbert spaces: the color Hilbert space and the helicity Hilbert space. Suppose the initial state carries  definite color and helicity labels, say
\begin{equation}\label{eq:ini_state_1}
  |\text{in}\rangle = \brr{|a_1\rangle \otimes |\lambda_1\rangle} \otimes \brr{|a_2\rangle \otimes |\lambda_2\rangle}  \ .
\end{equation}
Since the scattering amplitude in Yang-Mills can be decomposed into a product of the helicity-dependent part and the color-dependent part, one can express the scattering matrix as
\begin{equation}
    M_{\text{out},\text{in}} = M^{(\rm color)}_{a_3 a_4,a_1 a_2} M^{\rm(hel.)}_{\lambda_3\lambda_4,\lambda_1\lambda_2} \ .
\end{equation}
Applying this on the above initial state, one gets
\begin{equation}
|\text{out}\rangle = \brr{M^{(\rm color)}_{a_3 a_4,a_1 a_2} |a_1\rangle \otimes |a_2\rangle} \otimes \brr{M^{(\rm hel.)}_{\lambda_3 \lambda_4,\lambda_1 \lambda_2} |\lambda_1\rangle \otimes |\lambda_2\rangle} \ .
\end{equation} 
This indicates that while there can be entanglement between the two particles in the final state, color and helicity are always separable. The same conclusion also applies for an initial state with entanglement between particle 1 and 2 but no mixing between color and helicity, for example
\begin{equation}
    |\text{in}\rangle = \brr{ \sum_{i,j} \alpha_{ij}| a\rangle_{1i}\otimes| a\rangle_{2j}}\otimes
    \brr{ \sum_{k,l} \beta_{kl}| \lambda\rangle_{1k}\otimes| \lambda\rangle_{2l}} \ .
\end{equation}

Critically, if CKD were not satisfied,  color and helicity  in the scattering amplitude would not have been separable, and generally,  color and helicity would always have been mixed in the final state even with no mixing in the initial state. 
This can be seen from the following reduction only using the color Jacobi identity,
\begin{equation}     K\,\,\equiv\,\,M_4^{\text{full,tree}} =  A_s c_s + A_t c_t + A_u c_u = (A_s - A_u) c_s + (A_t - A_u) c_t .
\end{equation}
The coefficients of $c_s$ and $c_t$ are two independent parameters, as a function of kinematic parameters especially the helicities. Acting on the initial state Eq.(\ref{eq:ini_state_1}), the out state is given by
\begin{equation}
  |{\rm out}\rangle = \srr{(A_s - A_u) |\lambda_1 \rangle \otimes  |\lambda_2 \rangle} \otimes \srr{  c_s |a_1 \rangle \otimes  |a_2 \rangle } + 
  \srr{(A_t - A_u) |\lambda_1 \rangle \otimes  |\lambda_2 \rangle} \otimes \srr{  c_t |a_1 \rangle \otimes  |a_2 \rangle } \ .
\end{equation}
This out-state is generically inseparable for the color and helicity indices.


It should be emphasized that the separability discussed here is obtained from a top-down perspective. In contrast (and in a complementary way), the authors of \cite{Hu:2025jne} adopt a bottom-up perspective, where the presence of the swap operator in the scattering matrix reveals a symmetry structure given by the direct product of the flavor and spin symmetry groups, thereby ensuring the separability between helicity and flavor.

It is therefore thanks to the separability property that one can study  entanglement in the color and helicity Hilbert spaces individually. In the rest of this Section, we  study entanglement in color space. We choose $\theta=\pi/2$ ($t=u$) as the benchmark scattering angle: the normalized kernel then lies on a fixed ray in the invariant-operator space, leading to a group-only (dynamics-independent) value of the product-input peak entanglement $E_\star^{(N)}$. We will study this next.

\subsection{$\theta=\pi/2$ universality and a group-invariant $E_\star$ for ${\rm Adj}\!\otimes\!{\rm Adj}$  Scattering}\label{groupinvariant}

In this Section, we study the behavior of the entanglement at the scattering angle $\theta=\pi/2$, where the kinematics completely decouples.

For $\theta=\pi/2$ in the c.m.\ frame  one has $t=u$. Under the assumptions of: (i) tree level, (ii) locality and crossing (so $K$ is a sum of the three channels) and (iii) no explicit $dd$ color tensors beyond those generated by the
double-$f$ factors in Eq.~\eqref{eq:doublef-Jacobi}, one has
\begin{lemma}[Right-angle ray]
\label{lem:right-angle-ray}
At $\theta=\pi/2$ $(t=u)$, the normalized kernel $K(\theta=\pi/2)$ is unique up to an overall scale:
it lies on a \emph{single ray} in the space of SU($N$)-invariant operators on ${\rm Adj}\!\otimes\!{\rm Adj}$.
More precisely,
\begin{equation}
K(\theta=\pi/2)\ \propto\ 
\begin{cases}
\ \ \ \, \ \ \mathbb I+\mathbb S-6\,P_{\rm singlet}, & N=2,\\[4pt]
\ \ \ \, \ \ \alpha_1\,\mathbb I\;+\;\alpha_2\,\mathbb S\;+\;\alpha_3\,P_{\rm singlet}\;+\;\alpha_4\,(\mathbb D_t-\mathbb D_u)\;+\;\alpha_5\,(\mathbb D_u-\mathbb D_s), & N\ge3,
\end{cases}
\label{eq:right-angle-direction}
\end{equation}
with fixed \emph{ratios} $(\alpha_1:\alpha_2:\alpha_3:\alpha_4:\alpha_5)$ that do not depend on the dynamics (i.e.\ on the values of $A_{s,t,u}$ at $\theta=\pi/2$), only on group theory.
\end{lemma}
This follows immediately from Eq.~(\ref{eq:K-SU3-basis-after-KC}). Then one has the following

\begin{theorem}[Group-invariant peak entanglement from product inputs]
\label{thm:group-invariant-Estar}
Define
\[
E_\star^{(N)}\ :=\ \max_{\text{product }|u\rangle\otimes|v\rangle}\ 
E\!\left(\,{\rm normalized\ out\text{-}state\ of}\ K(\pi/2)\,|u\rangle\otimes|v\rangle\right),
\]
where $E$ is the linear entropy on the ${\rm Adj}$ bipartition.
Under assumptions (i)–(iii), $E_\star^{(N)}$ is independent of the dynamics and depends only on the gauge group,
i.e.\ $E_\star^{(N)}$ is a \emph{group invariant}.
\end{theorem}

\begin{proof}[Idea]
By Lemma~\ref{lem:right-angle-ray}, $K(\pi/2)$ is defined up to an overall scalar in a fixed direction
of the invariant operator space. After normalizing the out-state, $E$ depends only on that direction,
not on the overall scale. The maximization over product inputs
$|u\rangle\otimes|v\rangle$ depends only on how this fixed operator acts on rank-1 projectors in ${\rm Adj}$,
which is determined entirely by SU($N$) invariants (Casimirs and the contraction identities used above).
Hence the maximum is a function of $N$ alone.
\end{proof}

\subsection{SU(2): $E_\star$ calculation at $\theta = \pi/2$}
\label{subsec:SU2-right-angle}

 At $\theta=\pi/2$ ($t=u$), the SU(2) kernel collapses to 
\begin{equation}
K_{cd,ab}\brr{\dfrac{\pi}{2}}\ \propto\ \delta_{ac}\delta_{bd}+\delta_{ad}\delta_{bc}-2\,\delta_{ab}\delta_{cd}\,\,.
\label{eq:SU2-Heisenberg}
\end{equation}
 Acting on any product input $\alpha = u \otimes v$, i.e., $\alpha_{ab}=u_a v_b$ with $u,v\in\mathbb R^3$, the normalized out–state has color entanglement
\(
E(\pi/2;u\!\otimes\!v)=3/4
\)
which is a constant; hence, the product–class peak is universal:
\begin{equation}
E_\star^{(2)}\;=\;\max_{\text{product }u\otimes v}E\big(\pi/2\big)\;=\;\frac{3}{4}.
\label{eq:SU2-Estar}
\end{equation}
We give some more details of this computation below, and relegate an explicit check to Appendix~\ref{maxentsu2piover2}.

It should be noted that the adjoint representation is a real representation. Therefore we focus on the non-entangled initial state $u\otimes v$ with $u$ and $v$ are states with real coefficients up to a global phase difference. 
As can be seen, the kernel operator in Eq.~(\ref{eq:SU2-Heisenberg}) is invariant under the diagonal transformation $(O\otimes O)\otimes(O\otimes O)$. Therefore, one can always  rotate the two vectors to  $u' = (1,0,0)$ and $v'=(\cos\vartheta,\sin\vartheta,0)$ in some new basis. Here $\cos\vartheta = u\cdot v$ is the scalar product of $u$ and $v$, invariant under the transformation. Applying the scattering matrix to this initial state, one can obtain the final state matrix
\begin{equation}
    |\text{out} \rangle \sim 
\begin{pmatrix}
0 & \dfrac{\sin\vartheta}{2} & 0\\[4pt]
\dfrac{\sin\vartheta}{2} & -\cos\vartheta & 0\\[4pt]
0 & 0 & -\cos\vartheta
\end{pmatrix}
\end{equation}
Directly computing the entanglement power one can get the constant value $3/4$ (Appendix~\ref{maxentsu2piover2}).

\subsection{SU(3): $E_\star^{(3)}$ calculation at $\theta = \pi/2$}
\label{subsec:SU3-right-angle}

Determining $E_\star^{(3)}$ for the case of SU(3) involves computing the matrix $\mathcal{N}$. The analytical treatment is complicated and not very illuminating, and we opt instead for a numerical approach.
The $u$ and $v$ unit vectors are generated with equal probability on the unit sphere $S^7$ in $\mathbb{R}^8$. We scan over $10^5$ random pairs of $u$ and $v$. The distribution is displayed as in Fig.~\ref{fig:SU3_ent} on the plane of $E(\text{out})$ versus $\cos\vartheta=u\cdot v$. Within the sampled points, we find $E_*^{(3)}$ is around 0.9067. 
\begin{figure}
    \centering
    \includegraphics[width=0.6\linewidth]{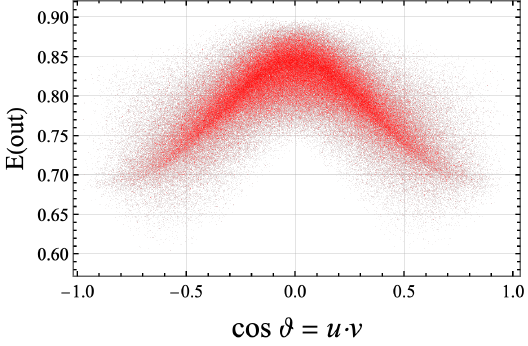}
    \caption{The sampling of points from on the plane of final state entanglement $E(\text{out})$ and the dot product of $u$ and $v$ vectors.}
    \label{fig:SU3_ent}
\end{figure}
 It can be seen that the entanglement is symmetric over $\cos\vartheta = 0$. This can be understood by adding a minus sign on the initial state matrix, namely, $u \otimes v \rightarrow -u\otimes v$, leading to $\cos\vartheta \rightarrow -\cos\vartheta$, which leaves the final state entanglement power  unchanged.

\subsection{Large N limit in SU(N): $E_*^{(N)}$ calculation at $\theta = \pi/2$}

We now consider the large $N$ limit. It is expected that the dominant contribution to the kernel is from the $d d$ sector. The general property is that each trace gives one factor of $N$. Therefore $\mathbb{D}_i \sim O(N)$ and the final state $\mathcal{N}$ matrix element is $\sim O(N)$. The entanglement $E\brr{|\text{out}\rangle}$ is given by
\begin{equation}
    E\ = \frac{N_A}{N_A-1}\!\left(1-\frac{\mathrm{Tr}\big[(\mathcal{N}\mathcal{N}^\dagger)^2\big]}
{\big(\mathrm{Tr}[\mathcal{N}\mathcal{N}^\dagger]\big)^2}\right)
\end{equation}
Counting the $N$ scalings of the numerator and denominator in the bracket, the ratio is always $\sim N^{-1}$. Hence as $N \rightarrow \infty$, the entanglement power approaches 1.

Recall that in the fundamental $\otimes$ fundamental scattering case, the entanglement power decreases with increasing $N$ and asymptotically approaches $1/2$ in the large-$N$ limit.
In contrast, for the $\rm Adj\,\otimes\,\rm Adj$ scattering, the entanglement power increases with $N$ and eventually approaches unity.
At first sight, this opposite trend is rather surprising.
A larger $N$ generally implies an enhanced underlying symmetry, which one would naturally expect to suppress the quantum entanglement generated by the scattering operator.
Hence, one might anticipate weaker entanglement for larger $N$.
However, our result shows the opposite behavior in the adjoint representation.

How do we understand this? We are starting from the initial two-particle states which live in $N\times N$ Hilbert space. For the fundamental representation, the states are in the fundamental of a SU(N) group, corresponding to $N^2 -1$ generators. However, it is possible, if one has no knowledge of the underlying group structure, that the same two-particle state which lives in $N\times N$ Hilbert space corresponds to the adjoint representation of a different group $SU(\sqrt{N+1})$. The invariant symmetry in this case only possesses N generators, which only accounts for $N/(N^2+1)$ part of all the possible generators.  In this sense, one can claim that the fundamental representation has ``higher symmetry" than the adjoint representation at fixed Hilbert space dimension.

\subsection{Probing EFT operators with color-space entanglement}\label{sec:EFT-updated}

We now turn to the question of effective operators and their effect on entanglement. Generally speaking, higher dimensional operators can introduce new color structures to the Yang-Mills Lagrangian. These extra color structures would induce a deformation away from the Yang-Mills locus in the basis $\{\mathbb I,\mathbb S,P_{\mathbf1},D_t-D_u,D_u-D_s\}$. It is cleanest to study such deformations at the scattering angle $\theta = \pi/2$. More details about Yang-Mills amplitudes are given in Appendix~\ref{ymscattbasics}.

We extend Yang-Mills by local gauge-invariant operators,
\begin{equation}
\mathcal L\;=\;\mathcal L_{\rm YM}\;+\;\frac{c_6}{\Lambda^2}\,\mathcal O_{F^3}
\;+\;\frac{1}{\Lambda^4}\sum_{j} c_{8,j}\,\mathcal O^{(8)}_j\;+\;\cdots,
\label{eq:EFT-L-upd}
\end{equation}
where
\(
\mathcal O_{F^3}\equiv f_{abc} F_{\mu}{}^{\nu\,a}F_{\nu}{}^{\rho\,b}F_{\rho}{}^{\mu\,c}
\)
is the unique CP-even three–field-strength operator at dim-6 (the $d^{abc}$ contraction vanishes for this Lorentz structure), and the dim-8 operators are listed later.
The induced kernel deformations can be parametrized as shifts of the five coefficients
\begin{equation}
\alpha_i=\alpha^{\rm YM}_i\;+\;\delta\alpha_i\,\,.
\label{eq:alpha-shifts-upd}
\end{equation}
At four points and tree level, dim-6 ($F^3$)\, operators have the same double-$f$ color span as Yang-Mills. A single $F^3$ 3-vertex joined to one Yang-Mills 3-vertex by a propagator produces channel color factors
\[
C_s=f_{abe}f_{cde},\quad C_t=f_{ace}f_{bde},\quad C_u=f_{ade}f_{bce},
\]
i.e. the same double-$f$ tensors as Yang-Mills exchange. Using
\(
f_{ace}f_{bde}
=\tfrac{2}{N}(\delta_{ab}\delta_{cd}-\delta_{ad}\delta_{bc})
+\big(d_{abe}d_{cde}-d_{ade}d_{cbe}\big),
\)
we see that $F^3$ populates both the permutation/singlet block \(\mathrm{span}\{\mathbb I,\mathbb S,P_{\mathbf1}\}\) and the  $dd$ block. Naively, one might expect that $\delta\alpha_{1,2,3}\big|_{F^3}\neq 0,$ and 
$\delta\alpha_{4,5}\big|_{F^3}\neq 0$.

However, the key point is that at dim-6, it has been shown that CKD still holds~\cite{Dixon:1993xd,Dixon:2004za,Broedel:2012rc}. In other words,  the scattering amplitude of YM plus dim-6 deformations can always be written in the general form,
\begin{eqnarray}
    M^{(6)} = M_s^{(6)} C_s + M_t^{(6)} C_t + M_u^{(6)} C_u\,\,
\end{eqnarray}
with $s M_s^{(6)} + t M_t^{(6)} + u M_u^{(6)} = 0$. One can always repeatedly perform the derivation as in Eq.~(\ref{Eq:A4_decomp_1}) to Eq.~(\ref{Eq:A4_decomp_2}) and get
\begin{equation}\label{Eq:A6_decomp}
M^{(6)} \sim  2\brr{M^{(6)}_t -M^{(6)}_s} \bigg[ t \, \text{Tr}\!\brr{ T_{a} T_{b} T_{c} T_{d}}   + u  \, \text{Tr}\!\left( T_{a} T_{b} T_{d} T_{c} \right) + s \, \text{Tr}\!\left( T_{a} T_{c} T_{b} T_{d} \right) \bigg]
\end{equation}
This has the same color structure as Yang-Mills except different kinematic amplitude. Moreover, the only non-vanishing helicity amplitudes are all plus/minus and single plus/minus helicity configurations.
The corresponding color-ordered amplitude is~\cite{Dixon:1993xd,Dixon:2004za,Broedel:2012rc}
\begin{align}
M^{(6)}(1^+,2^+,3^+,4^+)
&\sim 
   \frac{2 s\,t\,u}
        {\langle 12\rangle \langle 23\rangle \langle 34\rangle \langle 41\rangle} , \\[6pt]
M^{(6)}(1^-,2^+,3^+,4^+)
&\sim 
   \frac{- [23]^2 [34]^2 [42]^2}
        {[12][23][34][41]} \,\,,
\end{align}
with other helicity amplitudes  vanishing. These non-zero amplitudes are  orthogonal to the Yang-Mills helicity amplitudes. The key point is that upon normalization, the pre-factor $2\brr{M^{(6)}_t -M^{(6)}_s} $ will drop out from $E$, and therefore one obtains 
\begin{equation}
\delta\alpha_{1,2,3}\big|_{F^3} = 0, \quad 
\delta\alpha_{4,5}\big|_{F^3} = 0\,\,.
\end{equation}
\textit{Entanglement is blind to dim-6 deformations of Yang-Mills}.

At dim-8, there are 9 relevant operators~\cite{Murphy:2020rsh,Li:2020gnx,Corbett:2024yoy}
\begin{align}
\mathcal{O}^{(1)}_{F^4} &= (F_{\mu\nu}^A F^{A\mu\nu})(F_{\rho\sigma}^B F^{B\rho\sigma}) \\[4pt]
\mathcal{O}^{(2)}_{F^4} &= (F_{\mu\nu}^A \tilde{F}^{A\mu\nu})(F_{\rho\sigma}^B \tilde{F}^{B\rho\sigma}) \\[4pt]
\mathcal{O}^{(3)}_{F^4} &= (F_{\mu\nu}^A F^{B\mu\nu})(F_{\rho\sigma}^A F^{B\rho\sigma}) \\[4pt]
\mathcal{O}^{(4)}_{F^4} &= (F_{\mu\nu}^A \tilde{F}^{B\mu\nu})(F_{\rho\sigma}^A F^{B\rho\sigma}) \\[4pt]
\mathcal{O}^{(5)}_{F^4} &= (F_{\mu\nu}^A F^{A\mu\nu})(F_{\rho\sigma}^B \tilde{F}^{B\rho\sigma}) \\[4pt]
\mathcal{O}^{(6)}_{F^4} &= (F_{\mu\nu}^A F^{B\mu\nu})(F_{\rho\sigma}^A \tilde{F}^{B\rho\sigma}) \\[4pt]
\mathcal{O}^{(7)}_{F^4} &= d^{ABE} d^{CDE} (F_{\mu\nu}^A F^{B\mu\nu})(F_{\rho\sigma}^C F^{D\rho\sigma}) \\[4pt]
\mathcal{O}^{(8)}_{F^4} &= d^{ABE} d^{CDE} (F_{\mu\nu}^A \tilde{F}^{B\mu\nu})(F_{\rho\sigma}^C \tilde{F}^{D\rho\sigma}) \\[4pt]
\mathcal{O}^{(9)}_{F^4} &= d^{ABE} d^{CDE} (F_{\mu\nu}^A F^{B\mu\nu})(F_{\rho\sigma}^C \tilde{F}^{D\rho\sigma})
\end{align}
For the first six operators, the corresponding color structure is always proportional to 
\begin{equation}
     C_{cd,ab} \sim \delta_{ab} \delta_{cd} + \delta_{ac} \delta_{bd} +\delta_{ad} \delta_{bc} = N_A P_{\rm singlet} + \mathbb I +\mathbb S  \ .
\end{equation}
The last three operators, on the other hand,  introduce a $d d$ product structure. Therefore one can conclude that any linear combination of the nine dimension-eight operators would change the color structure of Yang-Mills, i.e., 
\begin{equation}
\delta\alpha_{1,2,3}\big|_{F^4}\neq 0, \quad 
\delta\alpha_{4,5}\big|_{F^4}\neq 0\,\,.
\end{equation}
The non-vanishing helicity amplitudes include both all plus/minus and MHV/$\overline{\rm MHV}$ configurations~\cite{Goldberg:2024eot}. Therefore it can interfere with either renormalizable part or dimension-6 operators. 
The entanglement is no longer determined solely by the color structure; it now depends on both the color structure and the kinematic amplitude.
At $\theta = \pi/2$, one can always make expansion over the Wilson coefficient and obtain
\begin{equation}
    E_{*\rm dim8}^{(N)} = \left\{
\begin{array}{ll}
E_{*}^{(N)} + \sum_{j}c_{8,j} \, C_j\,\dfrac{ s^2} {\Lambda^4} \ , & \text{MHV/}\overline{\text{MHV}} \\
E_{*}^{(N)} + \sum_{j} \dfrac{c_{8,j}}{c_6} \, D_j\,\dfrac{ s} {\Lambda^2} \ , & \text{all } +\!/{-}
\end{array}
\right.
\end{equation}
where $C_j$ and $D_j$ are the corresponding pre-factors.

\section{Entanglement in Helicity Space}\label{sec:entanle_hel}

Our focus up to this point has been on entanglement in color space. Now, we turn to the question of bipartite entanglement in the \emph{helicity} Hilbert space for $2\!\to\!2$ scattering of massless vectors (for related work, we refer to~\cite{Nunez:2025dch}). 
We will first introduce, in Section~\ref{sec:entanglement},  a feature of Yang-Mills: the fact that a maximally entangled (in helicity space) initial state always leads to \textbf{maximally entangled } final states. We will show the underlying algebraic structure of this statement in Section~\ref{subsec:algebraic-blocks}. The observable we will use is the linear entropy of the normalized out-state, which we define in Eq.~\eqref{eq:linentropy} (it is essentially the same entanglement measure that we have chosen till now, with the relevant dimension of the Hilbert space). Subsequently, in Section~\ref{deformedYM},  we will be interested in performing controlled deformations away from the Yang-Mills Lagrangian and checking the behavior of the entanglement.

\subsection{(Un-deformed) YM Satisfies Maximum Entanglement $\rightarrow$ Maximum Entanglement}
\label{sec:entanglement}

In this Section, we study an interesting property of the Yang-Mills Lagrangian
\begin{equation}
\mathcal L_{\rm YM} \;=\; -\frac{1}{4}\, F^a_{\mu\nu} F^{a\,\mu\nu},
\label{eq:L-YM}
\end{equation}
namely, that an initial state that is maximally entangled (MaxE) in helicity space remains maximally entangled after a 2$\to$2 scattering process (we will sometimes call this the ``MaxE $\to$ MaxE principle").

To compute the entanglement, the first step is to calculate the scattering matrix $M_{\lambda_3\lambda_4,\lambda_1,\lambda_2}$. 
In Sec.~\ref{YMkins2}, we have seen that the scattering amplitude can always be decomposed into a direct product of the kinematic part and the group structure part due to the CKD. As in Eq.~\eqref{eq:M4_full}, the full 4-point tree-level scattering amplitude can be written as
\begin{equation}\label{eq:M4_full}
\begin{split}
    M_4^{\rm full,tree} =&  g^2 M_4[1234]
    \bigg\{ \text{Tr}\!\left( T_a T_b T_c T_d \right) + \text{Tr}\!\left( T_a T_d T_c T_b \right) +\dfrac{u}{t} \big( \text{Tr}\!\left( T_a T_b T_d T_c \right) \\
    +&  \text{Tr}\!\left( T_a T_c T_d T_b \right) \big) + \dfrac{s}{t} \big( \text{Tr}\!\left( T_a T_c T_b T_d \right)
    +  \text{Tr}\!\left( T_a T_d T_b T_c \right) \big) \bigg\} \ .
\end{split}
\end{equation}
Therefore if there is no mixing between color and helicities in the initial states, the scattering matrix in the helicity space is purely determined by the color-ordered amplitude. The  Parke-Taylor formula at four point gives 
\begin{equation}\label{eq:parke_taylor}
M_{ij} \equiv M_4(\ldots i^-, \ldots j^-) 
= g^2 \frac{\langle ij \rangle^4}{\langle 12 \rangle \langle 23 \rangle \langle 34 \rangle \langle 41 \rangle}
\end{equation}
The amplitudes of all plus/minus and single plus/minus helicities are equal to zero. Only the two plus/minus helicities have non-vanishing amplitude. 
For different combinations of the helicity assignment, one can  replace the numerator of Eq.~\eqref{eq:parke_taylor} into  ket pairs with the corresponding external legs carrying minus helicity. In total there 6 nonvanishing configurations. Therefore the nonzero matrix components for (rescaled) $M_{\lambda_3\lambda_4,\lambda_1\lambda_2}$ are
\begin{equation}
\begin{split}
    M_{++,--} &= \langle 12\rangle^4, \quad
    M_{+-,+-} = \langle 24\rangle^4, \quad
    M_{+-,-+} = \langle 14\rangle^4,\\
    M_{-+,+-} &= \langle 23\rangle^4,\quad
    M_{-+,-+} = \langle 13\rangle^4,\quad
    M_{--,++} = \langle 34\rangle^4 \ .
\end{split}
\end{equation}
The alternative approach is to compute the scattering amplitude using Feynman rules. 
We may adopt the conventional polarization vectors
\begin{equation}
\epsilon^{\mu}(\hat{\mathbf{k}}, \pm ) 
= \frac{1}{\sqrt{2}} 
\Big(0,\ \mp \cos\theta \cos\varphi + i \sin\varphi,\ \mp \cos\theta \sin\varphi - i \cos\varphi,\ \pm \sin\theta \Big).
\label{eq:conventional-pol}
\end{equation}
Here $\hat{\mathbf{k}} = (\sin\theta\cos\varphi,\sin\theta\sin\varphi,\cos\theta)$ is the spatial unit vector of the external particle momentum with the scattering polar angle $\theta$ and the azimuthal angle $\varphi$. 
The azimuthal angle $\varphi$ introduces an additional phase in the scattering amplitude. 
The amplitude can be written as the conventional form,
\begin{equation}
  \mathcal{M} = \dfrac{n_s c_s}{s} +  \dfrac{n_t c_t}{t} + \dfrac{n_u c_u}{u}  \ .
\end{equation}
Since the entanglement in helicity space is not dependent on the group structure, we choose SU(2) group as the benchmark and $f_{abc} = \epsilon_{abc}$. 

We take the incoming two–qubit state to be maximally entangled. From App.~\ref{sec:setup}, a convenient parametrization is
\begin{equation}
|\mathrm{in}\rangle \;=\; U(\phi,\gamma,\psi)\,\frac{|+\!+\rangle+|-\!-\rangle}{\sqrt2},
\qquad U\in SU(2)\ .
\label{eq:MEinput}
\end{equation}
One can obtain
\begin{equation}\label{eq:in_state}
 |\text{in} \rangle  = 
 \begin{pmatrix}
     \cos\gamma + i \sin\gamma \cos\phi & i \sin\gamma \sin\phi \, e^{-i \psi} \\
     i \sin\gamma \sin\phi \, e^{i\psi} & \cos\gamma - i \sin\gamma \cos\phi 
 \end{pmatrix}  \ ,
\end{equation}
up to one global phase with the real angles $\gamma,\phi,\psi \in[0,2\pi]$.
We will scan over the angles ($\phi, \gamma, \psi$) in the following, to demonstrate that the results are independent of the initial state choice. 

Given an initial two-qubit helicity state $|\mathrm{in}\rangle$, the unnormalized out-state is
\begin{equation}
|{\rm out}\rangle_{\lambda_3\lambda_4} \;=\; M_{\lambda_3\lambda_4,\lambda_1\lambda_2}\,|{\rm in}\rangle_{\lambda_1\lambda_2}.
\end{equation}
We must normalize before computing entanglement:
\begin{equation}
\widetilde\rho \;\equiv\; \frac{|{\rm out}\rangle\langle{\rm out}|}{\langle{\rm out}|{\rm out}\rangle},
\qquad 
\widetilde\rho_R=\mathrm{Tr}_2\,\widetilde\rho,
\qquad
E \;=\; \frac{G}{G-1}\Big(1-\mathrm{Tr}\,\widetilde\rho_R^{\,2}\Big),
\qquad G=2.
\label{eq:linentropy}
\end{equation}

\begin{table}[!t]
\centering
\begin{tabular}{l>{\centering
\arraybackslash}p{0.2\linewidth}>{\centering\arraybackslash}p{0.2\linewidth}>{\centering\arraybackslash}p{0.2\linewidth}}
\toprule
 Parameter& Type & Symbol \\
\midrule
 scattering angle & kinematics & $\theta,\varphi$  \\
initial state & kinematics & $\phi, \gamma, \psi$\\
\midrule
\end{tabular}
\caption{\label{parameterstable}
The 5 parameters we choose to scan in YM, including  $\phi,\gamma,\psi$ parameterizing the initial maximally entangled state and the scattering angles $\theta$ and $\varphi$. 
}
\label{tab:entangparamsYM}
\end{table}

A scan of the five parameters in Tab.~\ref{tab:entangparamsYM} reveals that the  final state is \textbf{always maximally entangled}. In order to understand this result, we examine the algebra explicitly.

\subsubsection{Algebraic check of MaxE $\rightarrow$ MaxE (Helicity Blocks)}
\label{subsec:algebraic-blocks}

When using the Feynman rules for computation, the building blocks of the scattering amplitude include the external polarization vectors, the vertex coupling and the propagators. For  2$\to$2 scattering in Yang-Mills, one can show that flipping the helicities of all the external polarization vectors yields 
\begin{equation}
M[h_3,h_4,h_1,h_2]=M[-h_3,-h_4,-h_1,-h_2]^* \ .
\end{equation}
This pairs the $16$ helicity entries into $8$ conjugate pairs. Turning on each pair individually and acting on the maximally entangled  initial state
\begin{equation}\label{eq:param_max_ent}
|\mathrm{in}\rangle=
\left(
\begin{matrix}
\alpha&\beta^*\\[-2pt]-\beta&\alpha^*
\end{matrix}
\right)
\end{equation}
one can obtain the structure of the final state matrix:
\begin{align}
    M[+,+,+,+] & = M[-,-,-,-]^* = a_1, \quad 
    |\text{out} \rangle = \begin{pmatrix}
        a_1 \alpha & 0\\
        0 & a_1^* \alpha^*
    \end{pmatrix} \\
    M[+,+,+,-] & = M[-,-,-,+]^* = a_2, \quad
    |\text{out} \rangle = \begin{pmatrix}
        a_2 \beta^* & 0\\
        0 & - a_2^* \beta
    \end{pmatrix} \\
    M[+,+,-,-] & = M[-,-,+,+]^* = a_3, \quad
    |\text{out} \rangle = \begin{pmatrix}
        a_3 \alpha^* & 0\\
        0 &  a_3^* \alpha
    \end{pmatrix} \\
    M[+,+,-,+] & = M[-,-,+,-]^* = a_4, \quad
    |\text{out} \rangle = \begin{pmatrix}
        -a_4 \beta & 0\\
        0 &  a_4^* \beta^*
    \end{pmatrix} 
\end{align}
\begin{align}
    M[+,-,+,+] & = M[-,+,-,-]^* = a_5, \quad
    |\text{out} \rangle = \begin{pmatrix}
        0 & a_5 \alpha\\
        a_5^* \alpha^* & 0
    \end{pmatrix} \\
    M[+,-,-,+] & = M[-,+,+,-]^* = a_6, \quad
    |\text{out} \rangle = \begin{pmatrix}
        0 & a_6 \beta^*\\
        -a_6^* \beta & 0
    \end{pmatrix} \\
    M[+,-,-,-] & = M[-,+,+,+]^* = a_7, \quad
    |\text{out} \rangle = \begin{pmatrix}
        0 & a_7 \alpha^*\\
        a_7^* \alpha & 0
    \end{pmatrix} \\
     M[+,-,+,-] & = M[-,+,-,+]^* = a_8, \quad
    |\text{out} \rangle = \begin{pmatrix}
        0 & a_8 \beta^*\\
        -a_8^* \beta & 0
    \end{pmatrix} \ .
\end{align}
From Eq.~\eqref{eq:param_max_ent}, one finds that MaxE $\to$ MaxE requires that
\begin{equation}\label{eq:MaxE_condition}
    a_2 = a_4 = a_5 = a_7 = 0 \ ,
\end{equation}
while at least one non-vanishing pair in {$a_1,a_3,a_6,a_8$}
preserves the MaxE form. 
In Yang-Mills only the MHV/$\overline{\rm MHV}$ pairs are nonzero. 
Explicitly, we have the non-vanishing elements
\begin{equation}\label{eq:MHV_conf}
    M_{++,--} = (M_{--,++})^* = \mathcal{M}_1,\quad
    M_{+-,+-} = (M_{-+,-+})^* = \mathcal{M}_2,\quad
    M_{+-,-+} = (M_{-+,+-})^*= \mathcal{M}_3\ .
\end{equation}
Here, we have denoted the three amplitude components as $\mathcal{M}_1,\mathcal{M}_2,\mathcal{M}_3$ respectively. 
This can be understood by noting the reversed helicity polarization vectors are simply complex conjugation of each other.
Using Eq.~\eqref{eq:MHV_conf} and acting on the initial state as in~Eq.~\eqref{eq:in_state}, one finds that the entanglement of the final state is indeed 1. We call this the ``MaxE $\to$ MaxE principle".

\subsection{``Deformed" YM}\label{deformedYM}

In this Section, we perform controlled deformations away from the Yang-Mills Lagrangian, in order to study the behavior of the MaxE $\to$ MaxE principle. The idea of performing ``controlled deformations" requires clarification. Yang-Mills theory is the unique theory to describe massless vector fields; any deformations away from the cubic/quartic vertex or the Jacobi identity would lead to a failure of the quantization procedure and inconsistency of the group algebra structure (unphysical longitudinal modes would not be canceled by ghosts, leading to negative-norm modes; there would be no nilpotent BRST operator, etc.). Indeed, no consistent scattering amplitude would exist, so the question of studying the entangling properties of such an amplitude would be moot. 

Our calculations within such a deformed Yang-Mills setting should be seen as purely \textit{formal}, aimed at studying how well entanglement \textit{detects} such unphysical deformations. We will proceed in two steps. In the first step, in Section~\ref{wardonshell}, we will recapitulate the fact that the Yang-Mills Lagrangian is the only legitimate Lagrangian for massless vector fields, by appealing to on-shell Ward identities (this is well-known and follows, for example, \cite{Cheung:2017pzi}). The reason for adopting on-shell Ward identities as a proxy for the consistency of Yang-Mills is that this language is  particularly suited to studying the properties of the entanglement as one violates those identities.

\subsubsection{General Lagrangian for Gauge Fields in ``deformed" YM}

Our starting point is a set of $N$ abelian $U(1)$ fields $A^a$ with $a = 1, 2, \cdots N$. We want to ask: what is the general form of the renormalizable Lagrangian which contains the $N$ $U(1)$ fields, preserves the gauge symmetry, and perhaps even possesses an enhanced symmetry?   Schematically, one can write down the Lagrangian
\begin{equation}
\mathcal{L} \, = \, -\frac{1}{4} \sum_{a=1}^N F_{\mu \nu, a} F^{\mu \nu, a} \, + \, {\rm \{cubic \,\, interactions\}}\, + \, {\rm \{quartic \,\, interactions\}}\,\,\,.
\end{equation}
The kinetic term is just the summation of the kinetic terms of $N$ abelian fields. Generally, the cubic and quartic interactions will break the $U(1)^N$ symmetry of the kinetic term. It is convenient to parametrize the Lagrangian in the following way: 
\begin{equation}
\mathcal L \;=\; -\frac14\sum_a F^{a}_{\mu\nu}F^{a\,\mu\nu}
\;+\; \alpha\, f_{abc}\, A^{a}_{\mu}A^{b}_{\nu}\partial^{\mu}A^{c\,\nu}
\;+\; \kappa\,\alpha^2\, \mathcal O^{(4)}[A] ,
\label{eq:L-general}
\end{equation}
where $\alpha$ is a coupling, while $f_{abc}$ and $\kappa$ are real coefficients to be determined.

\subsubsection{On-shell Ward identities at Three and Four Points Force YM Locus}\label{wardonshell}

Our calculations in this Section  do not assume a non-Abelian gauge symmetry off shell. 
We allow the most general cubic $AA\partial A$ and a local quartic with an a priori free coefficient $\kappa$, and we only impose on-shell constraints. That is, we demand \emph{on-shell Ward identities} of the $S$-matrix: for each external leg $i$, a shift
\begin{equation}\label{eq:ward-onshell}
\varepsilon_i^\mu\;\to\; \varepsilon_i^\mu+\xi\,p_i^\mu
\end{equation}
must leave all on-shell amplitudes invariant. It is well known that demanding on-shell gauge invariance together with locality fixes the cubic coefficients $f_{abc}$ to be antisymmetric, determines the coefficient in the quartic term to be the Yang-Mills coefficient ($\kappa=1$), and enforces the  Jacobi relation among the coefficients $f_{abc}$ (we call this set of conditions the ``YM locus"). Equivalently, one may start from a non-Abelian gauge symmetry off shell, in which case these statements are automatic. By contrast, a genuine $U(1)^N$ gauge symmetry forbids nonzero $AA\partial A$ altogether.

One may ask why we begin with the on-shell Ward identities as opposed to the off-shell symmetry. If one postulates the off-shell symmetry, the result (antisymmetric $f_{abc}$, $\kappa\!=\!1$, Jacobi) is guaranteed and the entanglement analysis, or the introduction of the language of entanglement, are redundant. 

The on-shell Ward identities, in this regard, encode gauge redundancy only at the level of observable amplitudes:
$\varepsilon_i\!\to\!\varepsilon_i+\xi p_i$ leaves $\mathcal M$ invariant.  Unlike the off-shell route, we do not assume a Lagrangian gauge symmetry, a field basis, or BRST structure a priori. This makes the argument compatible with S-matrix bootstrap logic and robust under field redefinitions and gauge choices. The sketch of the proof is as follows.

\noindent \textit{Three points (antisymmetry)}: 
The most general three-vector on-shell amplitude is
\begin{equation}
\mathcal A_3 \;\propto\; f_{a_1a_2a_3}\Big[(\varepsilon_1\!\cdot\!\varepsilon_2)(\varepsilon_3\!\cdot\!(p_1-p_2))
+(\varepsilon_2\!\cdot\!\varepsilon_3)(\varepsilon_1\!\cdot\!(p_2-p_3))
+(\varepsilon_3\!\cdot\!\varepsilon_1)(\varepsilon_2\!\cdot\!(p_3-p_1))\Big].
\end{equation}
Enforcing Eq.~(\ref{eq:ward-onshell}) on each leg removes any symmetric part of $f_{a_1a_2a_3}$, leaving
\begin{equation}
f_{abc} \;=\; -\,f_{acb}\, .
\end{equation}

\noindent \textit{Four points $($YM quartic and Jacobi from on–shell Ward $)$}: 
At four points, the tree amplitude is
\begin{equation}
\mathcal M_4 \;=\; g^2\!\left(\frac{c_s\,n_s}{s}+\frac{c_t\,n_t}{t}+\frac{c_u\,n_u}{u}\right)
\;+\; g^2\,\kappa\,\mathcal C_4\, ,
\label{eq:4pt-decomp}
\end{equation}
where $c_s=f_{abe}f_{cde}$, $c_t=f_{ace}f_{bde}$, $c_u=f_{ade}f_{bce}$ are the color factors, $n_{s,t,u}$ are kinematic numerators, and $\mathcal C_4$ is the (local) contact tensor from $\mathcal O^{(4)}$.
We do not assume a Lie algebra a priori.

\medskip
\noindent\emph{Sketch of the on-shell Ward analysis.}
Choose leg~1 and replace $\varepsilon_1\to p_1$.
For the standard cubic vertex 
\(
V_{\mu\nu\rho}(p,q,r)=\eta_{\mu\nu}(p-q)_\rho+\eta_{\nu\rho}(q-r)_\mu+\eta_{\rho\mu}(r-p)_\nu
\),
one finds
\begin{equation}
p_1^\mu V_{\mu\nu\rho}(p_1,p_2,-k_s)\;=\;(k_s^2-p_2^2)\,\eta_{\nu\rho}\,,
\qquad
p_1^\mu V_{\mu\nu\rho}(p_1,p_3,-k_t)\;=\;(k_t^2-p_3^2)\,\eta_{\nu\rho}\,,
\label{eq:pV-identities}
\end{equation}
and analogously for the $u$-channel. With external legs on shell ($p_i^2=0$), the variations from the three exchange diagrams are \emph{local} and proportional to $+c_s$, $+c_t$, $+c_u$, respectively (the propagators $1/s$, $1/t$, $1/u$ cancel against $k^2=s,t,u$ from Eq.~(\ref{eq:pV-identities})). 
The quartic contact produces a local contribution proportional to $-\kappa\,(c_s+c_t+c_u)$ in the same tensor basis.
Hence the total variation under $\varepsilon_1\to p_1$ takes the form
\begin{equation}
\delta_1 \mathcal M_4 \;=\; (\kappa-1)\,\mathcal T_1 \;+\; (c_s+c_t+c_u)\,\mathcal T_2\,,
\end{equation}
where $\mathcal T_{1,2}$ are local tensors built from $\eta_{\mu\nu}$ and the remaining polarizations.
Repeating the argument on legs $2,3,4$ yields the same two independent tensor structures. 
Locality then forces
\begin{equation}
 \kappa=1, \,\qquad c_s+c_t+c_u=0 \, .
\label{eq:YM-conditions}
\end{equation}
The first condition fixes the quartic to the Yang–Mills value; the second is the color Jacobi identity at four points. 
Rephrasing, and parametrizing the deformations away from Jacobi explicitly, one may write 
\begin{equation}
    c_t\,=\, \chi_t c_s\qquad c_u=\chi_u c_s
\end{equation}
and then Eq.~\eqref{eq:YM-conditions} becomes $\chi_u=-1-\chi_t$.
Thus, on-shell gauge invariance alone (with locality) singles out the Yang-Mills locus at tree level.

\subsubsection{MaxE $\to$ MaxE Principle Equivalently Also Forces  YM Locus}

\begin{table}[!t]
\centering
\begin{tabular}{l>{\centering
\arraybackslash}p{0.2\linewidth}>{\centering\arraybackslash}p{0.2\linewidth}>{\centering\arraybackslash}p{0.2\linewidth}}
\toprule
 Parameter& Type & Symbol \\
\midrule
 scattering angle & kinematics & $\theta$  \\
initial state & kinematics & $\phi, \gamma, \psi$\\
color factor & gauge symmetry & $\chi_t, \chi_u$ \\
polarization & gauge symmetry & $\xi$\\
four-point interaction & gauge symmetry & $\kappa$\\
\midrule
\end{tabular}
\caption{\label{parameterstable}
The 8 parameters on which the entanglement depends, and their types.
}\label{entangparams}
\end{table}

In this Section, we investigate what happens to a maximally entangled initial state, if one allows deformations of Yang-Mills. At a given $\theta$, the entanglement depends on: the quartic deformation $\kappa$; the color-factor deformations $\chi_t$ and $\chi_u$, the three angles $(\phi,\gamma,\psi)$ that parametrize the MaxE input (a $U(2)$ rotation of a Bell state), and the \emph{unphysical} polarization–shift parameter $\xi$ on any chosen external leg,
\begin{equation}
\varepsilon^\mu\ \to\ \varepsilon^\mu+\xi\,p^\mu.
\label{eq:pol-shift}
\end{equation}
For demonstration, we choose to shift the the outgoing particle with momentum $p_3$. 
Thus one may list eight knobs $\{\theta,\kappa,\chi_t,\chi_u;\phi,\gamma,\psi;\xi\}$. It should be noted that we do not have the gauge redundancy for the full Lagrangian. While we are still choosing the external polarization vector as in Eq.~\eqref{eq:conventional-pol}, we add the additional $\xi$ parameter to choose different polarization vectors satisfying the on-shell transverse condition.

\subsubsection*{Fixed polarization ($\xi = 0$)}

We first choose $\xi = 0$, namely, the conventional polarization vector form. 
The quartic coefficient $\kappa$ and the color factors $(\chi_t,\chi_u)$ are varied. 
The amplititude can always be decomposed into
\begin{equation}
  \mathcal{M} \sim \dfrac{n_s}{s} + \dfrac{n_t \, \chi_t}{t} + \dfrac{n_u \, \chi_u}{t}
\end{equation}
in which $n_s,n_t,n_u$ are the same as that in the Yang-Mills case. It turns out that each kinematic factor $n_i (i=s,t,u)$ satisfies the MaxE$\rightarrow$MaxE condition Eq.~\eqref{eq:MaxE_condition} respectively. 
For a MaxE input, this implies
\begin{equation}
E(\theta;\kappa{=}1,\chi_t,\chi_u;\xi{=}0)=1
\qquad\text{for all }(\chi_t,\chi_u)\, .
\label{eq:Eeq1_fixedxi}
\end{equation}
On the other hand, the contact deformation repopulates the forbidden helicity blocks linearly in $(\kappa-1)$, and (after normalization) lowers $E$ quadratically:
\begin{equation}
E(\theta;\kappa,\chi_t,\chi_u;\xi{=}0)
=1-\mathcal C_\kappa(\theta;\chi_t,\chi_u)\,(\kappa-1)^2+\mathcal O\!\big((\kappa-1)^3\big),
\qquad \mathcal C_\kappa>0.
\label{eq:E_kappa_scan}
\end{equation}
Therefore a fixed-$\xi$ scan (like Fig.~\ref{fig:kappa_scan_fixedxi}) necessarily singles out $\kappa=1$ as the unique MaxE$\to$MaxE point, independently of whether Jacobi holds. This is why the Yang-Mills locus appears only as $\kappa=1$ in that figure. With $\kappa$ fixed to 1 and vanishing $\xi$, $E$ is identically 1 for any MaxE input and any $(\chi_t,\chi_u)$; scanning   $(\chi_t,\chi_u)$ would then therefore only check proper normalization and MHV selection. 

\subsubsection*{Fixed $\kappa = 1$}

The Jacobi identity is invisible unless we test the on-shell Ward shift. To obtain it, we will vary $\xi$ and demand $E$ be $\xi$-independent; the resulting gauge-variance heatmap vanishes exactly on $\chi_u = -1 - \chi_t$. We show this next.

Violations of the color Jacobi identity,
\begin{equation}
\delta J\;:=\;c_s+c_t+c_u \;=\; c_s(1+\chi_t+\chi_u)\,,
\label{eq:Jdef}
\end{equation}
do not alter the MHV selection for the fixed polarization basis $\xi = 0$.
\begin{figure}[t]
  \centering
  \begin{subfigure}[t]{0.88\textwidth}
    \centering
    \includegraphics[width=\linewidth]{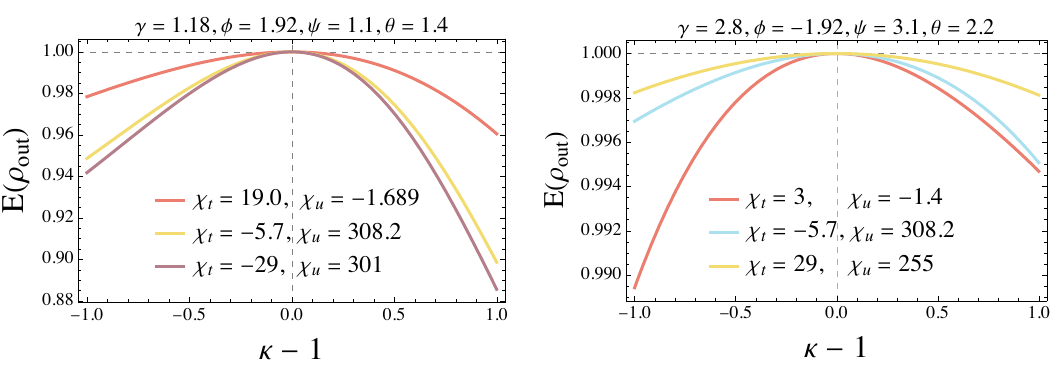}
    \caption{\textbf{Fixed polarization; $\kappa$ scan.}
    Linear-entropy entanglement $E(\theta;\kappa,\chi_t,\chi_u;\xi\!=\!0)$ vs.\ $\kappa$ for several random $(\chi_t,\chi_u)$ at fixed scattering angle $\theta$ (MaxE input).
    All curves peak at $E=1$ exactly at $\kappa=1$ and satisfy $1-E \sim \mathcal C_\kappa(\theta;\chi)\,(\kappa-1)^2$ near $\kappa=1$.
    This panel \emph{diagnoses the quartic} but does \emph{not} constrain Jacobi.\\}
    \label{fig:kappa_scan_fixedxi}
  \end{subfigure}

  \begin{subfigure}[h]{0.4\textwidth}
    \centering
    \includegraphics[width=\linewidth]{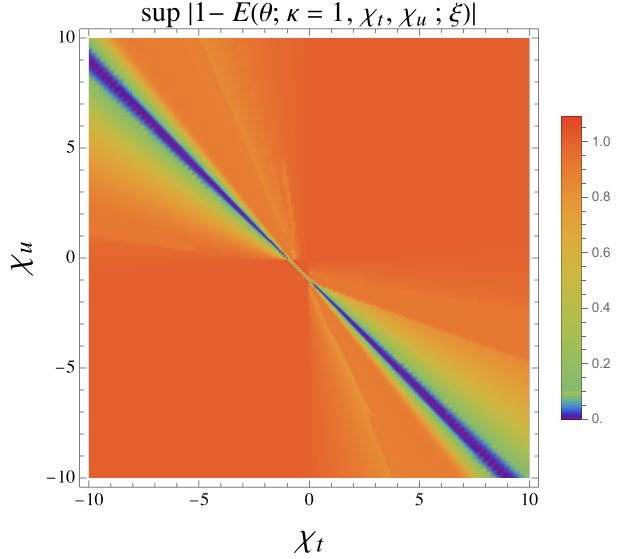}
    \caption{\textbf{On-shell Ward test; $\xi$-variance at $\kappa=1$.}
    Heatmap over $(\chi_t,\chi_u)$ of the gauge-variance metric $\Delta E_\xi(\chi_t,\chi_u)$. $\Delta E_\xi$ vanishes \emph{exactly} on the Jacobi line $\chi_u=-1-\chi_t$ and is positive elsewhere, scaling as
    $\Delta E_\xi \propto \xi_{\max}^2\,(1+\chi_t+\chi_u)^2$ near the locus.
    This panel \emph{diagnoses the Jacobi identity}.}
    \label{fig:xi_variance_heatmap}
  \end{subfigure}
\hspace{1cm}
%
  \begin{subfigure}[h]{0.4\textwidth}
    \centering
    \includegraphics[width=\linewidth]{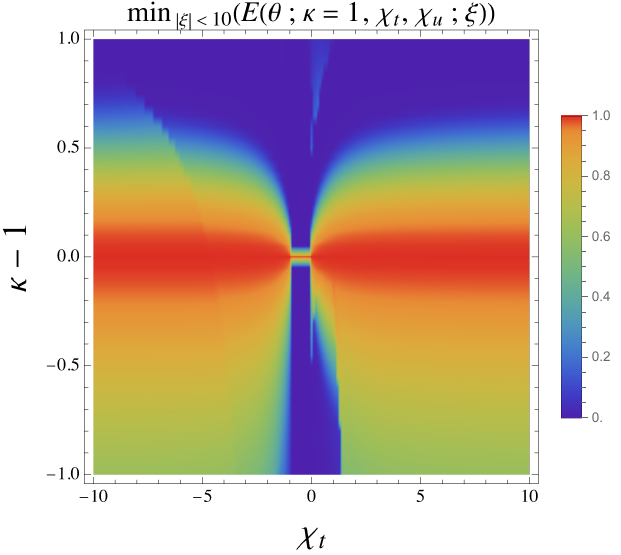}
    \caption{\textbf{Full YM locus}
    Defining $E_{\min}(\kappa;\chi_t,\chi_u):=\min_{|\xi|\le \xi_{\max}} E(\theta;\kappa,\chi_t,\chi_u;\xi)$,
    if $\chi_u=-1-\chi_t$ (Jacobi), then $E_{\min}$ attains $1$ at $\kappa=1$; off the Jacobi line, $E_{\min}(\kappa\!=\!1)<1$.
    This summarizes that the full YM locus requires \emph{both} $\kappa=1$ and Jacobi. We choose $\xi_{\rm max} = 10$. }
    \label{fig:Emin_vs_kappa}
  \end{subfigure}

  \caption{\textbf{Entanglement diagnostics of the Yang–Mills locus.}
  All panels use a maximally entangled (MaxE) input and fixed $\theta$ (no phase-space integration).
  (\textbf{A}) With fixed polarization ($\xi=0$), $E$ singles out $\kappa=1$ for any $(\chi_t,\chi_u)$: this is the quartic test.
  (\textbf{B}) Varying the on-shell polarization shift $\varepsilon\to\varepsilon+\xi\, p$ at $\kappa=1$ exposes Jacobi via $\Delta E_\xi$; the heatmap is dark \emph{only} on $\chi_u=-1-\chi_t$.
  (\textbf{C}) Worst-case over $\xi$ vs.\ $\kappa$ shows that the \emph{robust} MaxE$\to$MaxE point occurs iff both conditions hold, i.e.\ on the YM locus.}
  \label{fig:entanglement_diagnostics}
\end{figure}
However, the on-shell Ward shift on an external leg changes the amplitude by a local term proportional to $\delta J$. Acting on a MaxE input at $\kappa=1$, the induced change in the out-state is linear in $\xi\,\delta J$ and populates the forbidden helicity sector. After normalization, the entanglement drops quadratically:
\begin{equation}
E(\theta;\kappa{=}1,\chi_t,\chi_u;\xi)
\;=\;
1\;-\;\mathcal C_\xi(\theta)\,\xi^2\,\delta J^2 \;+\; \mathcal O(\xi^3)\,,
\qquad \mathcal C_\xi(\theta)>0\,.
\label{eq:E_xi_scan}
\end{equation}
Hence $E$ is independent of $\xi$ (the physical requirement) iff $\delta J=0$, i.e.\ iff the color Jacobi identity holds:
\begin{equation}
E(\theta;\kappa{=}1,\chi_t,\chi_u;\xi)\equiv1\ \ \text{for all }\xi
\quad\Longleftrightarrow\quad
 \chi_u=-1-\chi_t\ \ .
\label{eq:E_xi_equiv_J}
\end{equation}
In  Fig.~\ref{fig:xi_variance_heatmap}, we depict a heat map of $\Delta E_{\xi}$ on the plane of $(\chi_t,\chi_u)$, where 
\begin{equation}
      \Delta E_\xi(\chi_t,\chi_u)
      := \sup_{|\xi|\le \xi_{\max}}  \big|\,E(\theta;\kappa\!=\!1,\chi_t,\chi_u;\xi)-E(\theta;\kappa\!=\!1,\chi_t,\chi_u;0)\,\big|.
\end{equation}
We note that $\Delta E_\xi$ vanishes exactly on the Jacobi line $\chi_u=-1-\chi_t$ and is positive elsewhere, scaling as
    $\Delta E_\xi \propto \xi_{\max}^2\,(1+\chi_t+\chi_u)^2$ near the locus. This panel diagnoses the Jacobi identity.

\subsubsection*{All parameter scanning}

One can now scan all the parameters. For demonstration, we define 
\begin{equation}
E_{\min}(\kappa;\chi_t,\chi_u):=\min_{|\xi|\le \xi_{\max}} E(\theta;\kappa,\chi_t,\chi_u;\xi)\,\,.
\end{equation}
If $\chi_u=-1-\chi_t$ (Jacobi), then $E_{\min}$ attains $1$ at $\kappa=1$; off the Jacobi line, $E_{\min}(\kappa\!=\!1)<1$. In the bottom panel   of Fig.~\ref{fig:Emin_vs_kappa}, we show  $E_{\min}$ versus $(\kappa -1)$ with separate lines for ``Jacobi line" and ``off-Jacobi line". 
    
For a MaxE input and fixed $\theta$, the normalized entanglement admits the joint expansion
\begin{equation}
E(\theta;\kappa,\chi_t,\chi_u;\xi)
= 1 - \mathcal C_\kappa(\theta;\chi)\,(\kappa-1)^2
    - \mathcal C_\xi(\theta)\,\xi^2\,\delta J^2
    + \mathcal O\!\big((\kappa-1)^3,\,\xi^3,\,(\kappa-1)\xi^2,\,(\kappa-1)^2\xi\big),
\label{eq:E_joint_expansion}
\end{equation}
with $\delta J=c_s(1+\chi_t+\chi_u)$. 
The $(\kappa-1)^2$ term reflects repopulation of forbidden helicity blocks by the contact deformation, while the $\xi^2\delta J^2$ term reflects the on-shell Ward variation of the exchange kernel under a polarization shift.
Thus fixed-$\xi$ scans diagnose $\kappa$ via the first term, whereas $\xi$-invariance diagnoses Jacobi via the second.

We make a few comments about the scattering angle $\theta$. Our entanglement constraints are pointwise in the scattering angle $\theta$. For a maximally entangled input and generic kinematics (away from the forward/backward limits),
the normalized entanglement admits the small-parameter expansion around the YM point
\begin{eqnarray}
E(\theta;\kappa,\chi_t,\chi_u;\xi)
&=& 1
  - \underbrace{\mathcal C_\kappa(\theta;\chi_t,\chi_u)}_{>0}\,(\kappa-1)^2
  - \underbrace{\mathcal C_\xi(\theta)}_{>0}\,\xi^2\,\big(c_s{+}c_t{+}c_u\big)^2 \nonumber \\
  &+& \mathcal{O}\,\!\big((\kappa-1)^3,\,\xi^3,\,(\kappa-1)\xi^2,\,(\kappa-1)^2\xi\big).
\label{eq:theta-local}
\end{eqnarray}
Both coefficients $\mathcal C_\kappa(\theta;\chi)$ and $\mathcal C_\xi(\theta)$ are positive functions of $\theta$
(but otherwise unremarkable). The conditions selected by the entanglement test are therefore independent of~$\theta$: they hold at each angle separately. For this reason, our main diagnostic figures are shown at fixed $\theta$:
they isolate (i) the quartic test, $E(\kappa)$ at fixed~$\xi$, and (ii) the Jacobi test, $\xi$-invariance at $\kappa=1$,
without introducing an extra kinematic axis that only rescales $\mathcal C_\kappa$ and $\mathcal C_\xi$.

\section{Conclusion}\label{sec:conclusion}

The central message of this work is that the entanglement generated by scattering processes follows directly from representation theory.  
Treating the two-body $S$-matrix as an SU(N)-equivariant map disentangles group structure from dynamics:  
the tensor-product decomposition $R\!\otimes\!R'$ determines the algebra of invariant operators, which in turn fixes the qualitative entangling power of the scattering matrix.  
Physical details enter only through the coefficients multiplying these invariant tensors.

For particles in the fundamental representation, the invariant algebra $\mathrm{Span}\{\mathbb{I},\mathbb{S}\}$ is two-dimensional; only the identity and swap directions preserve separability, while generic combinations generate entanglement. 

Adjoint-adjoint scattering, by contrast, involves a five-dimensional invariant operator space $\{\mathbb{I},\mathbb{S},\mathbb{P}_{\text{singlet}},\mathbb D_t-\mathbb D_u,\ \mathbb D_u-\mathbb D_s\}$ whose additional projectors and $dd$ tensors inevitably correlate the two legs.  
This makes adjoint scattering intrinsically an \emph{entangler}, independent of whether the underlying $\mathrm{SU}(N)$ symmetry is global or gauge.

In Yang-Mills theory, CKD constrains the coefficients of these invariant tensors so that, at right angle ($\theta=\pi/2$), the normalized color kernel lies on a unique ray in the invariant-operator space.  
The resulting product-state entanglement $E_\star^{(N)}$ is a group invariant:  
for $SU(2)$ one finds $E_\star=3/4$, for $SU(3)$ numerically $E_\star\simeq0.91$, and as $N\!\to\!\infty$ the entanglement approaches unity.  

Dimension-six deformations ($F^3$ operators) respect this structure and leave $E_\star^{(N)}$ unchanged, whereas dimension-eight operators ($F^4$) populate new color sectors and produce calculable shifts of the universal value.  
Entanglement in color space thus functions as a tomographic probe of effective operators, separating those that preserve the gauge–theoretic structure from those that do not.

In helicity space, we studied the entanglement of final states with maximally entangled initial states and find the final states always remain maximally entangled. 
This remarkable feature originates from the MHV property of the Yang–Mills scattering amplitude. 
To further substantiate this observation, we introduce controlled deformations away from the Yang–Mills theory and perform a formal consistency check.
It is found that the entanglement of the final states decreases whenever the theory deviates from the Yang–Mills locus. The persistence of the MaxE-to-MaxE behavior thus emerges as a unique characteristic of Yang–Mills dynamics.

It should be noted that the above statements are only valid at the tree level. At the loop level, more color structures like the double trace would enter, which would introduce more involved kinematical dependence in  color space. Moreover, the MHV behavior is violated by the appearance of rational terms at one loop. 

 We chose Yang-Mills theory as the adjoin scattering benchmark  for calculations in this work, but our treatment would apply equally well to any quantum system whose Hilbert space carries an SU(N) representation, like scattering  of baryons and mesons,  adjoint Higgs theories, etc. We leave a detailed treatment of such theories for future work.

We note that in this work, we have studied the entanglement of out states for \textit{given} initial states, using the linear entropy as a basis independent measure of  entanglement. On the other hand, if one wants to study entanglement over \textit{all} input states, then one would have to average over the product of two complex projective spaces where the input states live, using the Fubini-Study metric. We have not considered this option in either color or helicity space in this work, leaving it for a future study.

Across both color and helicity sectors, entanglement emerges as a basis-independent diagnostic of gauge consistency and representation structure.  It distinguishes minimally and intrinsically entangling representations, isolates EFT deformations, and restates the Ward constraints that characterize Yang-Mills theory.  
These results suggest that the information-theoretic viewpoint unifies algebraic, geometric, and dynamical aspects of scattering.

\section*{Acknowledgement}

We thank Jean-Fran\c{c}ois Fortin, Zhen Liu, Ian Low and especially Navin McGinnis for useful discussions. 
K.F-L and KS would like to thank the Aspen Center for Physics, which is supported by National Science Foundation grant PHY-2210452, for hospitality during the course of this work. KS is supported in part by the National Science Foundation under Award No. PHY-2514896.

\appendix

\section{Entanglement Formalism}\label{sec:setup}

Let us consider a Hilbert space,  in  flavor space, helicity space, color space, or some other category of interest. The dimension of the Hilbert space is denoted as G, namely there are G basis states, $|1 \rangle, \cdots |G \rangle$. 
For example, the eigen-states of the helicity operator $h = \hat{p} \cdot \Vec{S}$  form the helicity Hilbert space. 
For a massless field with spin $S$, the two states are $|+S\rangle$ and $|-S \rangle$ which serve as the basis of the two-qubit state.

We focus on the scattering of two incoming particles and consider the product of the two Hilbert spaces $\mathcal{H}_1 \bigotimes \mathcal{H}_2$. A general state in this product space is
\begin{equation}
    |\alpha \rangle = \sum_{i,j} \alpha_{ij} |ij \rangle,\quad
    |ij \rangle = |i \rangle_1 \otimes |j\rangle_2 \ .
\end{equation}
Therefore the incoming state can be described by the $G\times G$ matrix $\alpha_{ij}$ with the normalization condition $\sum_{ij} \alpha^*_{ij} \alpha_{ij} = {\rm Tr} (\alpha^\dagger \alpha) = 1$. If the rank of the $\alpha$ matrix equals  1, then it can be decomposed into the direct sum of two vectors, which indicates the state to be an unentangled state. 

There are many ways to parametrize the entanglement of the state $\alpha$. One convenient measure is the linear entropy. For a bipartite system with density matrix $\rho$, the linear entropy is expressed as
\begin{equation}
    E(\rho) = \dfrac{G}{G - 1} |1 - {\rm Tr} \rho_R^2|  = \dfrac{G}{G-1} \bigg| 1- \dfrac{\text{Tr}\srr{(\alpha^\dagger \alpha)^2}}{(\text{Tr} \, \alpha^\dagger \alpha)^2}\bigg| \ , 
\end{equation}
in which $\rho_R$ is the reduced density matrix obtained by tracing out one particle. The coefficient is for the normalization so that the maximal entropy is equal to 1. The minimal entropy is equal to 0 which indicates no entanglement between the two particles in the Hilbert space. For a pure state  $\alpha$, one often use the notation $E(\alpha)$. 
For Hilbert spaces with $G = 2$ and $G=3$, one has 
\begin{align}
    E(\alpha)|_{G = 2} &= 4 \lambda_1 \lambda_2 \\
    E(\alpha)|_{G = 3} &= 3\brr{\lambda_1 \lambda_2 + \lambda_2 \lambda_3 + \lambda_1 \lambda_3} \ .
\end{align}
Here $\lambda_i$ are the eigenvalues of the matrix $\alpha^\dagger \alpha$ with the normalization $\sum_i \lambda_i = 1$. 
This implies that the states $\alpha, \alpha^*, \alpha^T$ and $\alpha^\dagger$ all have the same entanglement power. 
Furthermore, if we transform the $\alpha$ matrix to $\alpha' = W \alpha V^\dagger$ with $W,V \subset U(G)$ group, the entanglement for $\alpha$ and $\alpha'$ are the same. 
We now demonstrate the details for $G=2$.

\subsection*{G = 2}
If the state is not entangled, then one eigenvalue is equal to zero. The matrix $\alpha^\dagger \alpha$ can then be expressed as
\begin{equation}
    \alpha^\dagger \alpha = U_2^\dagger 
    \begin{pmatrix}
     1 & \\
     & 0
    \end{pmatrix} U_2 = U_2^\dagger 
    \begin{pmatrix}
     1 & \\
     & 0
    \end{pmatrix} U_1 U_1^\dagger 
    \begin{pmatrix}
     1 & \\
     & 0
    \end{pmatrix} U_2 \ ,
\end{equation}
where $U_1$ and $U_2$ are the group elements in the U(2) group. Therefore the general form of the matrix $\alpha$ with vanishing entanglement is equal to
\begin{equation}
    \alpha = U_1^\dagger \begin{pmatrix}
     1 & \\
     & 0
    \end{pmatrix}
     U_2 = \begin{bmatrix}
        (U_1)_{11}\\ (U_1)_{21} 
    \end{bmatrix} 
    \begin{bmatrix}
    (U_2)_{11} & (U_2)_{12}
    \end{bmatrix}
\end{equation}
which is indeed the direct product of two vectors.

For the case of the maximally entangled state, or $E(\alpha) = 1$, the only solution is $\lambda_1 = \lambda_2 = 1/2$. Then similarly one can conclude that $\alpha$ is just the group element of U(2). 
This is consistent with the definition of the concurrence for the two-qubit state. 
The general form is expressed in terms of the basis of the Pauli matrix
\begin{equation}
    \alpha = \dfrac{e^{i \beta}}{\sqrt{2}} \brr{\cos\frac{\gamma}{2} + i\, \boldsymbol{n} \cdot \boldsymbol{\sigma} \sin\frac{\gamma}{2}} \ .
\end{equation}
Here $\boldsymbol{n}$ is a unit vector in 3-dimensional space and $\boldsymbol{\sigma}$ is a Pauli matrix. Or one use the alternative expression
\begin{equation}
    \alpha = \dfrac{e^{i \beta}}{\sqrt{2}}
 \begin{pmatrix}
     \cos\gamma + i \sin\gamma \cos\tau & i \sin\gamma \sin\tau \, e^{-i \psi} \\
     i \sin\gamma \sin\tau \, e^{i\psi} & \cos\gamma - i \sin\gamma \cos\tau 
 \end{pmatrix}  \ .
\end{equation}
Another possible parameterization is 
\begin{equation}
    \alpha = \dfrac{e^{i \beta}}{\sqrt{2}} \begin{pmatrix}
        \psi & \gamma \\
        -\gamma^* & \psi^*        
    \end{pmatrix}
\end{equation}
with the constraint $|\gamma|^2+|\psi|^2 = 1$.

\section{Scattering Amplitudes in Yang-Mills Theory}\label{ymscattbasics}

In this Appendix, we recapitulate some well-known aspects of the  scattering amplitude for  Yang-Mills, and place them in the context of the main text. 
We study the 2-to-2 tree level scattering amplitude of the massless gauge fields in center-of-mass(c.m.) frame, $g_{a}(p_1,\lambda_1) + g_{b}(p_2,\lambda_2)\rightarrow g_{c}(p_3,\lambda_3) + g_{d}(p_4,\lambda_4)$ with $a_i$ the color index and $\lambda_i$ the helicity of the $i$-th particle. The momentum is assigned as following
\begin{align}
p_1 &= -(E, 0, 0, E), &
p_3 &= (E, E \sin\theta \cos\varphi, E \sin\theta \sin\varphi, E \cos\theta), \nonumber\\
p_2 &= -(E, 0, 0, -E), &
p_4 &= (E, -E \sin\theta \cos\varphi, -E \sin\theta \sin\varphi, -E \cos\theta),
\end{align}
so that all momentum is outgoing. We will use the spinor-helicity formalism. The massless momentum can be decomposed into
\begin{equation}
    p_{\alpha\dot{\alpha}} = p^\mu (\sigma_\mu)_{\alpha\dot{\alpha}} = \lambda_\alpha \tilde{\lambda}_{\dot{\alpha}} = |p\rangle_{\alpha} [ p|_{\dot{\alpha}} \ .
\end{equation}
For the general momentum $p =(E, E \sin\theta \cos\varphi, E \sin\theta \sin\varphi, E \cos\theta)$, the explicit forms are
\begin{equation}
    |p\rangle_\alpha = \sqrt{2E} \begin{pmatrix}
        \cos\frac{\theta}{2} \\
        \sin\frac{\theta}{2} e^{i\varphi}
    \end{pmatrix}, \quad
[p|_{\dot{\alpha}} = \sqrt{2E} \begin{pmatrix}
        \cos\frac{\theta}{2} \\
        \sin\frac{\theta}{2} e^{-i\varphi}
    \end{pmatrix}
\end{equation}
The spinor product is defined as
\begin{equation}
    \langle ij \rangle = \epsilon^{\alpha\beta} |i\rangle_\alpha |j\rangle_{\beta}, \quad  
     [ ij ]= \epsilon^{\dot{\alpha}\dot{\beta}}  
     | i ]_{\dot{\alpha}} | j ]_{\dot{\beta}} \ .
\end{equation}
Among various helicity configurations, only the two-plus-two minus helicity configuration has non-vanishing amplitude, which is called the MHV (Maximally Helicity Violation) amplitude. 
This behavior can be understood from the little group scaling and gauge invariance. 
To isolate the color structure factor and the kinematics, 
one  first exploits the identity $i f_{abc} = \mathrm{Tr}\!\left( T_a T_b T_c \right) - \mathrm{Tr}\!\left( T_b T_a T_c \right)$ and the completeness relation $\left( T_a \right)_i^{\ j} \left( T_a \right)_k^{\ l} 
= \brr{\delta_i^{\ l} \delta_k^{\ j} 
-\delta_i^{\ j} \delta_k^{\ l}/N}/2$ to convert the color structure constant product into
\begin{align}\label{eq:color_decompo}
    f^{a b e} f^{e c d} 
&\propto 
\mathrm{Tr}\!\left( T_{a} T_{b} T_c T_d \right)\,
- \,\mathrm{Tr}\!\left( T_a T_b T_d T_c \right)\nonumber \\
&- \mathrm{Tr}\!\left( T_a T_c T_d T_b \right)
\,+ \,\mathrm{Tr}\!\left( T_a T_d T_c T_b \right) \ .
\end{align}
Therefore the full tree level four-point amplitude simplifies to
\begin{equation}
    M_{4}^{\text{full,tree}}
= g^{2} \left( M_{4}[1234] \, 
\text{Tr}\!\left( T_{a} T_b T_c T_d \right) 
+ \text{perms of } (234) \right).
\end{equation}
Here the $M_4[1234]$ is the gauge invariant partial amplitude, or \emph{color-ordered amplitude}. One can use the kinematic terms in Feynman rules for computation. The  Parke-Taylor formula gives $n$-gluon tree level amplitude~\cite{Parke:1986gb,Elvang:2013cua,Cheung:2017pzi}
\begin{equation}
    M_{n}[1^{+}\,\ldots\, i^{-}\,\ldots\, j^{-}\,\ldots\, n^{+}]
= \frac{\langle i j \rangle^{4}}{\langle 1 2 \rangle \langle 2 3 \rangle \cdots \langle n 1 \rangle} \, .
\end{equation}
For the four point amplitude, we have
\begin{equation}
    M_{4}[1^{-}\,2^{-}\,3^{+}\,4^{+}]
= \frac{\langle 1 2 \rangle^{4}}{\langle 1 2 \rangle \langle 2 3 \rangle \langle 3 4 \rangle \langle 4 1 \rangle} \, .
\end{equation}
Color-ordered amplitudes have the following properties:
\begin{itemize}
    \item Cyclic: $M_4[1234] = M_4 [2341]$
    \item Reflection: $M_4[1234] = (-1)^4 M_4[4321]$
    \item $U(1)$ decoupling identity:
    $M_4[1234]+M_4[2134]+M_4[2314] = 0$
\end{itemize}
Therefore the six color-ordered amplitudes satisfy
\begin{equation}
    M_4[1432] = M_4[1234],\quad
    M_4[1342] = M_4[1243], \quad
    M_4[1423] = M_4[1324]
\end{equation}
The photon decoupling identity or the Kleiss-Kuijf relation gives
\begin{equation}
 M_4[1234] + M_4[1342] + M_4[1423] = 0   \ .
\end{equation}
We use the following notation
\begin{align}
    M_4[1234] &= M_4[1432] , \quad
    M_4[1243] = M_4[1342] , \nonumber \\
    M_4[1324] &= M_4[1423] = -(M_4[1234] + M_4[1243]) 
\end{align}
Furthermore, the BCJ relation~\cite{Bern:2008qj} gives
\begin{equation}
    s_{14} M_4[1234] - s_{13} M_4[1243] = 0 \quad \text{or} \quad
    M_4[1234] \, u = M_4[1243] \, t 
\end{equation}
This reduces the number of the independent color-ordered amplitudes to one. The total amplitude can then be written as
\begin{equation}\label{eq:M4_full}
\begin{split}
    M_4^{\rm full,tree} =&  M_4[1234]
    \bigg\{ \text{Tr}\!\left( T_a T_b T_c T_d \right) + \text{Tr}\!\left( T_a T_d T_c T_b \right) +\dfrac{u}{t} \big( \text{Tr}\!\left( T_a T_b T_d T_c \right) \\
    +&  \text{Tr}\!\left( T_a T_c T_d T_b \right) \big) + \dfrac{s}{t} \big( \text{Tr}\!\left( T_a T_c T_b T_d \right)
    +  \text{Tr}\!\left( T_a T_d T_b T_c \right) \big) \bigg\}
\end{split}
\end{equation}
An alternative parametrization is used in the main text:
%
\begin{equation}
  M_4^{\text{full,tree}} = A_s\,c_s+\;A_t\,c_t+A_u\,c_u   \ ,
\end{equation}
in which $c_s=f_{abe}f_{cde},
c_t=-f_{ace}f_{bde},
c_u=f_{ade}f_{bce}$
are the corresponding color factors. 
They satisfy the Jacobi identity $c_s + c_t + c_u = 0$. 
The numerator $n_i$ factors depend on the kinematics and the external polarization vectors. 
From color-kinematics duality, one can always find a gauge choice such that 
\begin{equation}\label{eq:kin_jacobi}
   s A_s + t A_t + u A_u = 0 \ .
\end{equation}
 One can decompose the color factor $c_i$ via Eq.~(\ref{eq:color_decompo}) to obtain
\begin{equation}\label{Eq:A4_decomp_1}
\begin{split}
    M_4^{\text{full,tree}} \sim & \brr{A_s - A_u} \srr{\text{Tr}\brr{T_a T_b T_c T_d}+ \text{Tr}\brr{T_a T_d T_c T_b}} + \\ 
    & \brr{A_t - A_s} \srr{\text{Tr}\brr{T_a T_b T_d T_c}+ \text{Tr}\brr{T_a T_c T_d T_b}} + \\
    & \brr{A_u - A_t} \srr{\text{Tr}\brr{T_a T_c T_b T_d}+ \text{Tr}\brr{T_a T_d T_b T_c}} \ .
\end{split}
\end{equation}
We denote $A_s - A_t = a(\Theta)$ to match the convention in the context.  
From Eq.~(\ref{eq:kin_jacobi}), one can get $A_u = -(s A_s + t A_t)/u$ and
\begin{eqnarray}
    A_s - A_u = -\dfrac{t}{u} a 
    \quad A_u - A_t = -\dfrac{s}{u} a 
\end{eqnarray}
Actually even we are not starting from the basis or gauge choice in which the BCJ relation does not apply, one can always do the ``generalized gauge transformation". The coefficients $A_i - A_j$ is invariant under such transformation.  Inserting into the above equation gives
\begin{equation}
\begin{split}
    M_4^{\text{full,tree}} \sim & a \bigg\{ t\srr{ \text{Tr}\!\brr{ T_{a} T_{b} T_{c} T_{d}} + \text{Tr}\!\brr{ T_{a} T_{d} T_{c} T_{b} } } + u \Big[ \text{Tr}\!\left( T_{a} T_{b} T_{d} T_{c} \right) \\
    +&  \text{Tr}\!\left( T_{a} T_{c} T_{d} T_{b} \right) \Big] + s \srr{ \text{Tr}\!\left( T_{a} T_{c} T_{b} T_{d} \right)
    +  \text{Tr}\!\left( T_{a} T_{d} T_{b} T_{c} \right) } \bigg\}
\end{split}
\end{equation}
which shares the same structure as Eq.~(\ref{eq:M4_full}). Further  simplification can be realized by the reflection property $\text{Tr}\brr{T_a T_b T_c T_d} = \text{Tr}\brr{T_a T_d T_c T_b}$. This can be clearly seen by noting that each non-Abelian Lie group always admits a real adjoint representation. As a result, we have
\begin{equation}\label{Eq:A4_decomp_2}
M_4^{\text{full,tree}} \sim  2a \bigg[ t \, \text{Tr}\!\brr{ T_{a} T_{b} T_{c} T_{d}}   + u  \, \text{Tr}\!\left( T_{a} T_{b} T_{d} T_{c} \right) + s \, \text{Tr}\!\left( T_{a} T_{c} T_{b} T_{d} \right) \bigg]
\end{equation}
%


\section{Entanglement in SU(2) Color Space for $\theta = \pi/2$} \label{maxentsu2piover2}

For SU(2), the structure constants  $f_{abc}=\epsilon_{abc}$ and $d_{abc} = 0$. Therefore one has
\begin{equation}
\begin{split}
    M_{cd,ab}(\theta) &\sim 2 \delta_{ab} \delta_{cd} - (1-\cos\theta)\delta_{ac} \delta_{bd} - (1+\cos\theta) \delta_{ad} \delta_{bc}  \\
    & = 2 P_\text{singlet} - (1-\cos\theta) \mathbb{I} -(1+\cos\theta) \mathbb{S}
\end{split}
\end{equation}
The explicit matrix form is as
\begin{equation}\label{eq:color_matrix}
    \begin{pmatrix}
\begin{pmatrix}
0&0&0\\
0&-1&0\\
0&0&-1
\end{pmatrix} &
\begin{pmatrix}
0&\sq&0\\
\cq&0&0\\
0&0&0
\end{pmatrix} &
\begin{pmatrix}
0&0&\sq\\
0&0&0\\
\cq &0&0
\end{pmatrix}
\\[6pt]
\begin{pmatrix}
0&\cq&0\\
\sq&0&0\\
0&0&0
\end{pmatrix} &
\begin{pmatrix}
-1&0&0\\
0&0&0\\
0&0&-1
\end{pmatrix} &
\begin{pmatrix}
0&0&0\\
0&0&\sq\\
0&\cq&0
\end{pmatrix}
\\[6pt]
\begin{pmatrix}
0&0&\cq\\
0&0&0\\
\sq&0&0
\end{pmatrix} &
\begin{pmatrix}
0&0&0\\
0&0&\cq\\
0&\sq&0
\end{pmatrix} &
\begin{pmatrix}
-1&0&0\\
0&-1&0\\
0&0&0
\end{pmatrix}
\end{pmatrix}
\end{equation}
Here we have dropped the common factor which would be normalized out when computing the density matrix. 
The entanglement function does not rely on the helicity configurations since there is only one independent color-ordered partial amplitude after imposing the Kleiss-Kuijf
relations and BCJ relations as shown in the last section.

\subsection{Non-entangled Initial State}

We first start from an initial state with zero entanglement. The initial state matrix can be parametrized as
\begin{equation}
    \alpha = U_1^\dagger \begin{pmatrix}
     1 & &\\
     & 0 &\\
     & & 0
    \end{pmatrix}
    \begin{pmatrix}
     1 & &\\
     & 0 &\\
     & & 0
    \end{pmatrix} U_2 = \begin{bmatrix}
        (U_1)_{11}\\ (U_1)_{21} \\ (U_1)_{31}
    \end{bmatrix} 
    \begin{bmatrix}
    (U_2)_{11} & (U_2)_{12} & (U_2)_{13}
    \end{bmatrix}
\end{equation}
For simplicity we select $U_1$ and $U_2$ from the subgroup SO(3). The general form of the SO(3) group elements is given by
\begin{equation}
\begin{split}
& R(\lambda_1,\lambda_2,\lambda_3) =  R_z(\lambda_1) R_y(\lambda_2) R_z(\lambda_3) =  \\  & \left(
\begin{array}{c@{\hspace{1.2em}}c@{\hspace{1.2em}}c}
-\sin\lambda_1 \sin\lambda_3+\cos\lambda_1 \cos\lambda_2 \cos\lambda_3 &
-\sin\lambda_1 \cos\lambda_3 - \cos\lambda_1 \cos\lambda_2 \sin\lambda_3 &
\cos\lambda_1 \sin\lambda_2
\\[6pt]
\sin\lambda_1 \cos\lambda_2 \cos\lambda_3 + \cos\lambda_1 \sin\lambda_3 &
-\sin\lambda_1 \cos\lambda_2 \sin\lambda_3 + \cos\lambda_1 \cos\lambda_3 &
\sin\lambda_1 \sin\lambda_2
\\[6pt]
-\sin\lambda_2 \cos\lambda_3 &
\sin\lambda_2 \sin\lambda_3 &
\cos\lambda_2
\end{array}
\right)\ .
\end{split}
\end{equation}
In the following we choose three different initial states as
\begin{equation}
    \alpha_1 = \mathbb{I}_3 \Lambda_1 R(\pi,\pi,0), \quad
    \alpha_2 = \Lambda_1 R\brr{0.57,1.1,2.25},
    \quad
    \alpha_3 = R\brr{\pi,\dfrac{\pi}{2},\dfrac{\pi}{3} } \Lambda_1 R\brr{\dfrac{\pi}{2},\dfrac{\pi}{4},\dfrac{\pi}{6}}
\end{equation}
in which $\Lambda_1$ is the diagonal matrix $ \text{Diag}\brr{1,0,0}$. 
\begin{figure}
    \centering
    \includegraphics[width=0.6\linewidth]{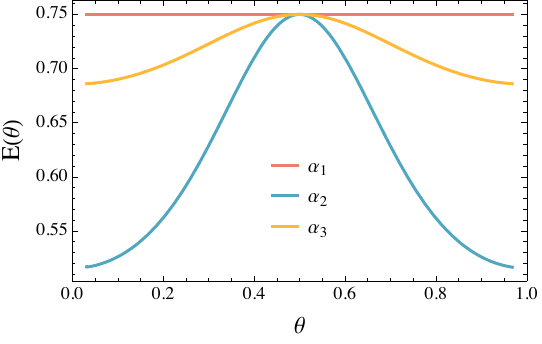}
    \caption{The entanglement as the function of the scattering polar angle $\theta$ for different 3 initial states.}
    \label{fig:entangle_theta}
\end{figure}

Applying the scattering matrix in color space on any initial non-entangled state, one finds that the entanglement of the final state reaches a maximal value of $3/4$ at $\theta = \pi/2$. 
In Fig.~\ref{fig:entangle_theta}, we display the entanglement as a function of the scattering polar angle $\theta$. 
As the initial states vary, the corresponding entanglement power profile with $\theta$ varies as well, with all curves meeting at the common point $(\theta = \pi/2, E = 3/4)$.

\bibliographystyle{utphys}
\bibliography{ref}

\end{document}